\documentclass[twoside]{article}

\usepackage[accepted]{aistats2024}
\usepackage[toc,page,header]{appendix}
\usepackage{minitoc}

\usepackage[round]{natbib}

\usepackage[utf8]{inputenc} 
\usepackage[T1]{fontenc}    
\usepackage[hyphens]{url}            
\usepackage{booktabs}       
\usepackage{amsfonts}       
\usepackage{nicefrac}       
\usepackage{microtype}      
\usepackage{xcolor}         
\usepackage{lipsum}
\usepackage{multirow}
\usepackage{algorithm,algorithmic}
\usepackage{natbib}
\usepackage{adjustbox}

\usepackage{bm}
\usepackage{bbm}
\usepackage{amsmath}
\usepackage{amssymb}
\usepackage{amsthm}
\usepackage{amsfonts} 
\usepackage{mathtools}
\usepackage{mathrsfs}

\usepackage{caption}
\usepackage{csquotes}
\usepackage{nicefrac}  
\usepackage{url}
\usepackage{shortcuts}
\usepackage{makecell}
\usepackage{tikz}
\usetikzlibrary{positioning}

\usepackage{enumitem}
\usepackage{thmtools,thm-restate}
\usepackage{cancel}
\usepackage{tikz}
\usepackage{pifont}

\usepackage{hyperref}       
\usepackage[capitalize]{cleveref}

\hypersetup{colorlinks,
			linkcolor=linkred,
			citecolor=citegreen}

\theoremstyle{plain}

\allowdisplaybreaks

\begin{document}

\doparttoc 


\usetikzlibrary{arrows}

%
\runningtitle{Benchmarking Observational Studies with Experimental Data under Right-Censoring}

%
\runningauthor{Demirel, De Brouwer, Hussain, Oberst, Philippakis, Sontag}

\twocolumn[

\aistatstitle{Benchmarking Observational Studies with \\ Experimental Data under Right-Censoring}

\aistatsauthor{Ilker Demirel \And Edward De Brouwer \And  Zeshan  Hussain}

\aistatsaddress{MIT \And  Yale University \And MIT}

\aistatsauthor{Michael Oberst \And Anthony Philippakis \And  David Sontag}

\aistatsaddress{Carnegie Mellon University \And  Broad Institute of MIT and Harvard \And MIT}
]

\begin{abstract}
    Drawing causal inferences from observational studies (OS) requires unverifiable validity assumptions; however, one can {\em falsify} those assumptions by benchmarking the OS with experimental data from a randomized controlled trial (RCT). A major limitation of existing procedures is not accounting for {\em censoring}, despite the abundance of RCTs and OSes that report right-censored time-to-event outcomes. We consider two cases where censoring time (1) is independent of time-to-event and (2) depends on time-to-event the same way in OS and RCT. For the former, we adopt a censoring-doubly-robust signal for the conditional average treatment effect (CATE) to facilitate an equivalence test of CATEs in OS and RCT, which serves as a proxy for testing if the validity assumptions hold. For the latter, we show that the same test can still be used even though unbiased CATE estimation may not be possible. We verify the effectiveness of our censoring-aware tests via semi-synthetic experiments and analyze RCT and OS data from the Women's Health Initiative study.
\end{abstract}

\section{INTRODUCTION}
The ability to reliably establish causal relationships is essential for decision-making and policy development \citep{pearl2018book, angrist2009mostly, Hernan2021-yd}. Although experimental data, often collected through randomized controlled trials (RCT), are considered to be the \enquote{gold-standard} for inferring causality, real-world evidence collected from observational (non-experimental) data is increasingly guiding regulatory processes \citep{hansford2023reporting, govc2019, au2022, nice2022uk}. Indeed, observational data, such as claims and electronic health records, can provide large-scale, diverse, and longitudinal data at low cost, making it a promising complement to time and cost-intensive experimental data. 

Ideally, one would like to leverage observational studies (OS) when experimental data is unavailable or provides limited evidence \citep{dagan2021bnt162b2, gershman2018using}. For instance, people with a history of cardiovascular diseases may not be eligible to participate in an RCT. Therefore, OSes are the only source of data for those people. Furthermore, the limited sample size of RCTs makes subgroup analysis infeasible and their results do not always apply to a {\em target} population of interest \citep{rothwell2005external, stuart2011use, Hartman2015-td, colnet2020causal}. Contrary to the RCTs, however, OSes are susceptible to numerous sources of bias, bringing their utility in practice under question. Therefore, it is critical to evaluate the credibility of an OS before using it for different downstream tasks \citep{yang2023elastic}. To that end, we will develop a hypothesis test to check if the findings from an OS and an RCT are compatible within the trial-eligible population, that can be used when the outcomes are {\em right-censored} \citep{kalbfleisch2011statistical}.

Naive analysis of an OS can lead to biased effect estimates due to various reasons. Among those, unobserved confounding--- which makes prognostic factors systemically differ in treatment and control groups--- typically receives the most attention. However, the bias may also emerge due to the poor analysis of the data regarding the handling of censored outcomes and non-adherence to treatment assignments \citep{hernan2017per}, the definition of time-zero and follow-up \citep{lodi2019effect, hernan2008observational}, and different types of {\em selection bias} \citep{hernan2004structural, yadav2021immortal}. Target trial emulation (TTE), where one uses observational data to emulate a hypothetical trial, has emerged as a popular framework to limit the bias in the OSes \citep{hernan2016using, wang2023emulation}. 

With a well-specified TTE protocol, it is possible to estimate causal effects from observational data under well-known {\em internal} and {\em external} validity assumptions (given in \Cref{sec:notset}) \citep{imbens2015causal, wager2018estimation, semenova2021debiased}. Internal validity ensures that the causal effects can be reliably inferred in the {\em OS population}. External validity further allows transporting those effect estimates to different populations ({\em e.g.}, the RCT population) \citep{dahabreh2020extending}. Even though these assumptions are not {\em verifiable}, one can still {\em falsify} them by benchmarking the OS to an RCT \citep{dahabreh2020benchmarking, forbes2020benchmarking}. The key idea is to formulate and test a null hypothesis that captures the implications of those assumptions, which is the equivalence of the treatment effects inferred from the OS and the RCT. The rejection of the null would then be linked to the violation of (a subset of) those assumptions.

Recent works have developed tests for falsifying the internal and external validity assumptions. \citet{hussain2022falsification} proposed an algorithm to first compare the {\em group-level} effects derived from an RCT and multiple OSes in pre-specified groups and integrate the evidence only from the OSes compatible with the RCT. \citet{hussain2023falsification} developed a falsification framework that compares {\em individual-level} effect estimates from an RCT and an OS over the entire covariate space and {\em automatically} detects the regions of disparity, providing explanations in the form of witness functions. \citet{de2023hidden} adopt an alternative view and focus on quantifying the hidden confounding in an OS from a \enquote{sensitivity analysis} perspective instead of testing for the equivalences of the effect estimates across studies. \citet{karlsson2024detecting} show how one can detect hidden confounding when there is no RCT data but {\em multiple} OSes that share a data-generating graph with certain properties. None of the studies above consider censored observations.
\paragraph{Our Contributions}
Censoring due to drop-outs or loss-to-follow-ups is a common issue that plagues both OSes and RCTs. Improperly handling the censored data does not merely lead to a suboptimal falsification test for the validity assumptions but renders the test unreliable since censoring can easily introduce bias. We generalize the test in \citet{hussain2023falsification} to cases with right-censored time-to-event outcomes. We first consider in \Cref{ssec:fwic} the common scenario where the censoring time is conditionally independent of the time-to-event. In \Cref{ssec:fwgc}, we introduce a novel censoring concept, {\em global censoring}, where the censoring time depends on time-to-event ({\em e.g.}, drop-out due to disease progression), but in the same way in the RCT and the OS. We develop a falsification test under both censoring mechanisms and verify its effectiveness in semi-synthetic experiments with the Infant Health and Development Program cohort. We also analyze real-world RCT and OS data from the Women's Health Initiative. \Cref{fig:1} gives an overview of our results.
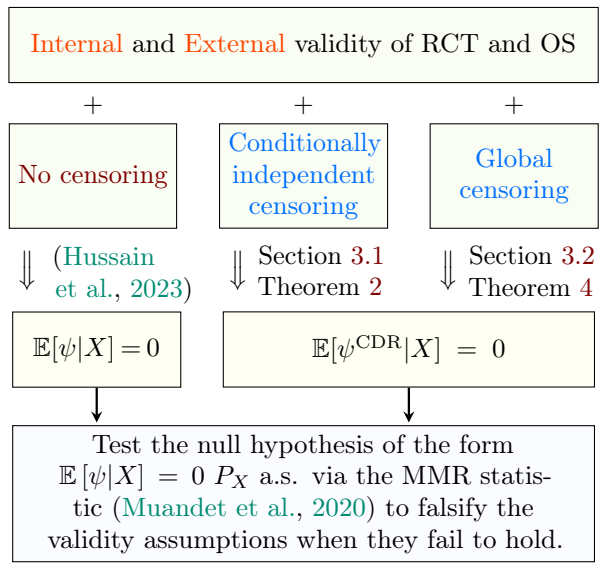
\begin{figure}[t]
\centering
\begin{tikzpicture}[
    box/.style={rectangle, draw, minimum size=1cm, align=center},
    arr/.style={->, >=stealth, shorten >=1pt, thick},
    node distance=1cm
]
\node[box, text width=7.6cm, fill=mygreen!10, minimum width=1cm] (box1-1) {\hyperref[asm:iva]{\myorange{Internal}} and \hyperref[asm:evo]{\myorange{External}} validity of RCT and OS};

\node[yshift = -.8cm, xshift=-2.8cm] (plus-1) {$+$};

\node[below=of plus-1, box, yshift = 1cm, text width=1.99cm, minimum height=1.39cm, fill=mygreen!10] (box3-1) {\maroon{No censoring}};

\node[xshift=-3.7cm, yshift=-3cm] (implies-1) {\rotatebox{270}{$\implies$}};
\node[below=of box3-1, yshift=0.9cm, xshift=0.5cm, text width=2cm] (impt-1) {\citep{hussain2023falsification}};

\node[below=of box3-1, draw, box, align=center, text width=2cm, xshift=0.05cm, yshift=-.1cm, fill=mypink!10] (box4-1) {$\!\mathbb{E} [ \psi \lvert X]\!=\!0\!$};


\node[yshift = -.8cm, xshift=0cm] (plus-2) {$+$};

\node[below=of plus-2, box, yshift = 1cm, text width=2cm, fill=mygreen!10] (box3-2) {\myorange{\hyperref[asm:cic]{\cblue{Conditionally}}} \\ \hyperref[asm:cic]{\cblue{independent}} \\  \hyperref[asm:cic]{\cblue{censoring}}};

\node[xshift=-0.9cm, yshift=-3cm] (implies-2) {\rotatebox{270}{$\implies$}};
\node[below=of box3-2, yshift=0.9cm, xshift=0.4cm, text width=2cm] (impt-1) {\Cref{ssec:fwic}\\ \Cref{thm:prethmMMR}};

\node[below=of box3-2, draw, box, align=center, text width=4.7cm, yshift=-.1cm, xshift=1.4cm, fill=mypink!10] (box4-2) {$\!\mathbb{E} [ \psi^{\tn{CDR}} \lvert X]\!=\!0\!$};

\node[yshift = -.8cm, xshift=2.8cm] (plus-3) {$+$};

\node[below=of plus-3, box, yshift = 1cm, text width=2cm, minimum height=1.39cm, fill=mygreen!10] (box3-3) {\mgray{\hyperref[asm:gc]{\cblue{Global}}} \\ \mgray{\hyperref[asm:gc]{\cblue{censoring}}}};

\node[xshift=1.9cm, yshift=-3cm] (implies-3) {\rotatebox{270}{$\implies$}};
\node[below=of box3-3, yshift=0.9cm, xshift=0.4cm, text width=2cm] (impt-1) {\Cref{ssec:fwgc}\\ \Cref{thm:gcThm}};

\node[below=of box4-2, box,xshift=-1.4cm, yshift =0.5cm, text width=7.5cm, fill=myblue!10] (box5-1) {Test the null hypothesis of the form $\E{\psi \lvert X} = 0$ $P_X$ a.s. via the MMR statistic \citep{muandet2020kernel} to falsify the validity assumptions when they fail to hold.};

\draw[arr, line width=0.3mm] (box4-2.south) -- (box4-2.south |- box5-1.north);
\draw[arr, line width=0.3mm] (box4-1.south) -- (box4-1.south |- box5-1.north);

\end{tikzpicture}
\caption{Prior work develops a maximum moment restriction (MMR)-based falsification test for validity assumptions, under no censoring. We extend the test to the case where time-to-event outcomes are right-censored, considering two censoring mechanisms.}
\label{fig:1}
\rule{\linewidth}{.75pt} 
\end{figure}
\section{NOTATION AND BACKGROUND} \label{sec:notset}
\paragraph{Notation} 
Let $A \in \cbrc{0,1}$ and $Y \in \mathbb{R}_+$ denote the binary treatment assignment and time-to-event outcome. We use $Y(a)$ to refer to the {\em potential outcome} for treatment $A = a$ and use $S$ as the study indicator with $S=0$ reserved for the RCT and $S=1$ for the OS. We denote by ${\cal I}_0$ and ${\cal I}_1$ the set of patients for the RCT and the OS, respectively, and let ${\cal I} \coloneqq {\cal I}_0 \cup {\cal I}_1$. We denote the cardinality of a set by $|{\cal I}|$ and let $|{\cal I}_0| = n_0$, $|{\cal I}_1| = n_1$, and $|{\cal I}| = n$, where $n = n_0 + n_1$.

We denote the set of patient covariates by $X$ and define ${\cal X}$ as the space of trial-eligible patients, that is,
\begin{equation*} 
    \prb{S = 0 \mid X = x} > 0, \qquad \forall x \in {\cal X}.
\end{equation*}

\paragraph{Internal and External Validity}
We first introduce conditional average treatment effect (CATE) estimation without censoring. For a patient with covariates $x \in {\cal X}$ in study $S = s$, the CATE is defined as the expected difference between the \textit{potential outcomes} with and without treatment:
\begin{align}
    \tn{CATE} (x, s) \coloneqq \E{Y(1) - Y(0) \lvert X=x, S = s}.
    \label{eq:cate}
\end{align}
The fundamental challenge in causal inference is that the potential outcomes $Y(0)$ and $Y(1)$ are never observed together. For a patient in the control group ({\em i.e.}, $A=0$) we observe $Y(0)$ but not $Y(1)$, and vice versa for the treatment group. Nevertheless, one can still estimate the CATE in a given study $S = s$ under certain internal validity assumptions listed below. 

\begin{restatable}[\textit{Internal validity}]{dfn}{ivr}
\label{dfn:ivr}
We say that internal validity holds for a study $S = s$ if the following conditions hold $\forall a \in \cbrc{0,1}$ and $\forall x \in {\cal X}$.
    \hspace{-1pt}
    \begin{itemize}[leftmargin=*]
        \item \textit{No unobserved confounding ---}  
            $Y(a) \indep A \mid X, S=s$.
        \item \textit{Consistency ---} 
            $A=a \implies Y = Y(a)$.
        \item  \textit{Positivity of treatment assignment}
        \begin{equation*}
        \prb{A=a \mid X=x, S=s} > 0.
        \end{equation*}
    \end{itemize}
\end{restatable}

\begin{restatable}[]{asm}{ivr}
    \label{asm:iva}
    Internal validity (\cref{dfn:ivr}) holds for both the RCT $S=0$ and the OS $S=1$.
\end{restatable}

\Cref{asm:iva} allows to identify the CATE in \eqref{eq:cate} as a quantity that can be estimated from data ({\em i.e.}, an estimand) in the RCT population and the OS population {\em separately}. However, even when the internal validity holds, the CATE in two populations can differ due to different distribution of treatment effect modifiers that are not included in $X$. To generalize the CATE from one study to the other, one needs external validity that assumes away unobserved treatment effect modifiers \citep{dahabreh2019, dahabreh2020extending}.

\begin{restatable}[\textit{External validity}]{asm}{evo}
    \label{asm:evo}
    We have, $\forall a \in \cbrc{0,1}$ and $\forall x \in {\cal X}$,
    \begin{itemize}[leftmargin=*]
        \item \textit{Ignorability of selection ---} $Y(a) \indep S \mid X$.
        \item \textit{Positivity of selection ---} $ \prb{S = 1 \mid X = x} > 0$.
    \end{itemize}
\end{restatable}

\paragraph{Falsification Without Censoring}
\cref{asm:iva} implies that $\tn{CATE}(x, 1)$ and $\tn{CATE}(x, 0)$ can be estimated, and \cref{asm:evo} implies that $\tn{CATE}(x, 1) = \tn{CATE}(x, 0)$. Intuitively, disagreement between the estimated CATE functions from each study implies that one or more of these assumptions are violated. Testing this equivalence forms the basis of the maximum moment restriction (MMR)-based falsification framework proposed in \citet{hussain2023falsification} in the {\em absence of censoring}. The core technical idea is to relate the equivalence of the underlying CATE functions to a set of conditional moment restrictions, finding an \enquote{instance-wise signal} $\psi$ such that Assumptions~\ref{asm:iva}-\ref{asm:evo} imply $\E{\psi \mid X} = 0$ almost surely. This reduction allows for applying the recent advances in the testing of MMRs via kernel methods \citep{muandet2020kernel}.
\section{CENSORED FALSIFICATION} \label{sec:prelim}
This work considers the common scenario where the time-to-event outcome $Y \in \mathbb{R}_+$ is subject to right-censoring. Let $C \in \mathbb{R}_+$ be the censoring time. For a patient, we observe $(X, A, S, \tilde{Y}, \Delta)$ where
\begin{equation} \label{eq:ytilde_def}
    \tilde{Y} = \min \para{Y,C}, \qquad \Delta = \ind{Y \leq C}.
\end{equation}
We either observe the time-to-event or the censoring time, indicated by $\Delta$. $\Delta = 1$ means that the time-to-event $Y$ is observed, but not the censoring time $C$, and vice versa for $\Delta = 0$. Censoring introduces an additional identification problem as neither of the potential outcomes is observed for censored patients. Without any assumptions on the censoring mechanism, unbiased CATE estimation is not possible in either study, rendering prior falsification approaches ineffective.

In this section, we generalize the falsification framework in \citet{hussain2023falsification} under two different censoring conditions. In \Cref{ssec:fwic}, we assume the censoring time $C$ is independent of the time-to-event $Y$ after conditioning on covariates $X$. We derive unbiased instance-wise signals for CATE in the RCT and the OS, suitable for comparison via an MMR test. Note that the falsification could be due to violating the censoring assumption, even though the validity assumptions hold. Remarkably, even when the censoring is not conditionally independent of time-to-event, and therefore unbiased CATE estimation is infeasible, testing the validity assumptions using the same signals may still be possible. The key ingredient is an alternative assumption that the censoring mechanism is identical across studies (\cref{ssec:fwgc}).
\subsection{Falsification with Conditionally Independent Censoring} \label{ssec:fwic}
We start by recalling the conditionally independent censoring condition, which is common in survival analysis~\citep{kalbfleisch2011statistical}.

\begin{restatable}[\textit{Conditionally independent censoring}]{asm}{cic} \label{asm:cic}
\[
    Y \indep C \mid X, A, S.
\]
\end{restatable} 

We show that under conditionally independent censoring, one can adapt the doubly-robust censoring-unbiased estimator in \citet{rubin2007doubly} to form conditional moment restrictions (CMR) that enables us to use an MMR-based test \citep{muandet2020kernel} for falsifying the validity assumptions (\ref{asm:iva}-\ref{asm:evo}). Precisely speaking, we derive instance-wise signals $\psi_{s}$ such that $\E{\psi_{s} \mid X} = \tn{CATE} (X,s)$, allowing us to define a difference signal $\psi = \psi_{1} - \psi_{0}$ such that $\E{\psi \mid X} = 0$ almost surely in $P_X$. We can then use the preceding equality as our null hypothesis, whose rejection would imply the violation of the validity assumptions, and use an MMR-based test similar to that of \citet{hussain2023falsification} (see~\cref{thm:thmMMR}).

By $G$ and $F$, we denote the cumulative distribution functions of the censoring time $C$ and time-to-event outcome $Y$, respectively. We use $\bar{G}$ and $\bar{F}$ to denote their {\em survival} functions ($\bar{G} (t) = 1 - G(t)$). For conciseness, we use the following notation
\begin{align}
    \bar{G}_{s,a}(t \mid X) \coloneqq \prb{C > t \mid X,S=s,A=a}. \label{eq:gbar}\\
    \bar{F}_{s,a}(t \mid X) \coloneqq \prb{Y > t \mid X,S=s,A=a}. \label{eq:fbar}
\end{align}
We further assume that for any realizable time-to-event, there is a nonzero probability of being observed, to prevent censoring-related identifiability issues.

\begin{restatable}[\textit{Support under censoring}]{asm}{ooo}
    \label{asm:ooo}
    $\forall t \in \mathbb{R}_+$,  $\forall a,s \in \cbrc{0,1}$, the following holds almost surely in $P_X$\footnote{All (in)equalities involving random variables hold almost surely throughout the manuscript.}    
    \[
        \bar{F}_{s,a}(t \mid X) > 0 \implies \bar{G}_{s,a}(t \mid X) > 0.
    \]
\end{restatable}

Combined with internal validity in \Cref{asm:iva}, Assumptions \ref{asm:cic} and \ref{asm:ooo} make it possible to have unbiased estimates of the CATE (see \eqref{eq:cate}) in the RCT and the OS populations in the presence of right-censoring.

We start by writing the censoring-unbiased signal from \citet{rubin2007doubly} for a study $s$ and treatment group $a$ pair.
\begin{align}
    \psi^{*}_{s, a} &= \frac{\ind{\Delta=1} Y}{\bar{G}_{s,a}(Y \mid X)} + \frac{\ind{\Delta=0} Q_{s, a}(X,C)}{\bar{G}_{s,a}(C \mid X)} \nonumber \\
    & \quad - \int_{-\infty}^{\tilde{Y}} \frac{Q_{s, a}(X, c)}{\bar{G}_{s,a}^2(c \mid X)} dG_{s,a}(c \mid X), \label{eq:def_rubin_signal}
\end{align}
where $Q_{s, a}(X, C) = \E{Y \lvert X, Y>C, S=s, A=a}$. $\psi^*_{s, a}$ is \enquote{doubly-robust} in the sense that it is unbiased for the time-to-event outcome $Y$ if either $\bar{G}_{s,a} (t|X)$ or $\bar{F}_{s,a} (t|X)$ is correctly estimated. Computing $\psi^*_{s,a}$ requires estimating the survival function of the censoring time, $\bar{G}_{s, a}(t \lvert X)$, and of the time-to-event outcome, $\bar{F}_{s, a}(t \lvert X)$. Using the estimate for $\bar{F}_{s, a}(t \lvert X)$, one can also calculate $Q_{s, a}(X, C)$ by integration under conditionally independent censoring. The practitioner may use covariate-adjusted Kaplan-Meier estimators or the Cox proportional hazards (CoxPH) framework to model the effect of covariates on the time-to-event along with recent advances in survival modeling \citep{van2001estimation, cole2004adjusted, cox1972regression, chapfuwa2021enabling, curth2021survite}. We adopt a CoxPH model and provide the details in \cref{sec:semiexp}. 

\begin{restatable}[]{lem}{lemDR}
    \label{lem:lemDR}
    Suppose that Assumptions~\ref{asm:iva},\ref{asm:cic}, and \ref{asm:ooo} hold. 
    From Theorem 1 in \cite{rubin2007doubly} and \Cref{asm:iva}, we have, $\forall s,a \in \cbrc{0,1}$,
    \[
    \E{\psi^*_{s,a} \mid X,S=s,A=a} = \E{Y(a) \mid X,S=s},
    \]
    if $\bar{G}_{s, a} (t \lvert X)$ \eqref{eq:gbar} or $\bar{F}_{s, a} (t \lvert X)$ \eqref{eq:fbar} is correctly estimated.\end{restatable}

Lemma~\ref{lem:lemDR} is quite powerful as it identifies conditional average potential outcomes under right-censoring in doubly-robust way. However, $\psi^*_{s,a}$ is not readily available as an {\em instance-wise} signal to facilitate an MMR-based test, as it is only defined for $S=s$ and $A=a$. One can handle this by re-weighting with inverse selection $P(S \lvert X)$ and propensity scores $P(A \lvert X,S)$.

\begin{restatable}[]{cor}{corDR}
    \label{cor:corDR}
    Suppose that Assumptions~\ref{asm:iva},\ref{asm:cic},\ref{asm:ooo} hold and $P(S=1 \lvert X)$ is correctly estimated. Let
    \begin{align}
    \psi_{s,a}^{\tn{IPW},*} \coloneqq \frac{\ind{S=s,A=a} \psi^*_{s,a}}{P(S=s \mid X) P(A=a \mid X, S=s)}.
    \end{align}
    Then $\forall s,a \in \cbrc{0,1}$, 
    $ 
    \E{\psi_{s,a}^{\tn{IPW},*} \lvert X} = \E{Y(a) \lvert X, S=s}
    $
    if $P(A=a \lvert X,S=s)$, and either $\bar{G}_{s, a} (t \lvert X)$ \eqref{eq:gbar} or $\bar{F}_{s, a} (t \lvert X)$ \eqref{eq:fbar} are correctly estimated.
\end{restatable}

Following \Cref{cor:corDR}, one can define the instance-wise signal $\psi_s^{\tn{IPW},*} \coloneqq \psi^{\tn{IPW},*}_{s,a=1} - \psi^{\tn{IPW},*}_{s,a=0}$ which is unbiased for the CATE in study $s$, that is, $\mathbb{E} [\psi_s^{\tn{IPW},*} \lvert X] = \tn{CATE}(X,s)$ (see \eqref{eq:cate}). We can then test the equivalence CATE($X,1$) $=$ CATE($X,0$) via an MMR test with the null hypothesis of $\mathbb{E} [\psi_{s=1}^{\tn{IPW},*} - \psi_{s=0}^{\tn{IPW},*} \lvert X] = 0$. 

\paragraph{Enhancing Double-robustness}
While we are now well-equipped for an MMR-based falsification test, \Cref{cor:corDR} requires the correct estimation of the propensity score of treatment $P(A=a \lvert X, S=s)$. We alleviate this requirement by building upon the doubly-robust estimation of treatment effects literature, where the correct estimation of either the propensity score or the mean outcome function is sufficient. In particular, since $\bar{F}_{s,a} (t \lvert X)$ entirely describes the time-to-event outcome distributions, $P(Y(a) \lvert X, S=s)$, estimating it correctly would suffice for estimating the CATE (see \eqref{eq:cate}), even when the propensity score estimation is incorrect. We propose the following censoring-doubly-robust (CDR) signal, which enjoys this enhanced doubly-robust property:
\begin{align} 
    \psi_{s,a}^{\tn{CDR}} &\coloneqq \frac{\ind{S=s}}{P(S=s \mid X)} \nonumber \\
    &\kern-1em \times \left( \frac{\ind{A=a} (\psi^{*}_{s, a} - \mu_{s, a}(X))}{P(A=a \mid X,S=s)} + \mu_{s, a}(X) \right), \label{eq:dr_signal}
\end{align}
where $\mu_{s, a} (X) \coloneqq \E{Y \lvert X, S=s, A=a}$ can be computed by integrating $\bar{F}_{s,a}(t \lvert X)$ over $t \in \mathbb{R}_+$.  

\begin{restatable}[]{thm}{thmDRtwo}
    \label{thm:thmDRtwo}
    Suppose that Assumptions~\ref{asm:iva}–\ref{asm:ooo} hold and $P(S=1 \lvert X)$ is correctly estimated. Then, $\forall s,a \in \cbrc{0,1}$
    \begin{equation*}
        \E{ \psi_{s,a}^{\tn{CDR}} \mid X} = \E{Y(a) \mid X, S=s},
    \end{equation*}
    if either $\bar{F}_{s,a} (t \lvert X)$ \eqref{eq:fbar}, or both $\bar{G}_{s,a} (t \lvert X)$ \eqref{eq:gbar} and $P(A=a \lvert X, S=s)$ are correctly estimated.
\end{restatable}

Although correctly estimating $\bar{F}$ under censoring remains challenging, the doubly-robust property is still desirable. For instance, consider a scenario where the censoring time tends to be high, and most of the observations are non-censored. In that case, due to the small number of censored observations, our estimates for $\bar{G}_{s,a} (t|X)$ may suffer from high variance, whereas the estimates for $\bar{F}_{s,a}(t|X)$ may be more reliable. It is also sufficient to correctly estimate $\bar{G}_{s,a} (t|X)$ and $P(A=a \lvert X, S=s)$. This is advantageous in scenarios where further assumptions on censoring simplify the estimation of $\bar{G}$ ({\em e.g.}, under type-1 censoring discussed at the end of this section). We investigate the doubly-robustness of our signal empirically in \Cref{supssec:drtest}.

Starting from \eqref{eq:dr_signal}, we define the censoring-doubly-robust instance-wise signals for CATE in study $s$ and \enquote{CATE difference} across studies as follows.
\begin{equation}\label{eq:cdrSig}
    \begin{split}
            \psi_{s}^{\tn{CDR}} &\coloneqq \psi_{s,a=1}^{\tn{CDR}} - \psi_{s,a=0}^{\tn{CDR}}.  \\
    \psi^{\tn{CDR}} &\coloneqq \psi_{s=1}^{\tn{CDR}} - \psi_{s=0}^{\tn{CDR}}.
    \end{split}
\end{equation}

\begin{restatable}{thm}{prethmMMR}
    \label{thm:prethmMMR}
    Suppose that Assumptions~\ref{asm:iva}–\ref{asm:ooo} hold and $P(S=1 \lvert X)$ is correctly estimated. Then
    \[
        \E{\psi_{s=1}^{\tn{CDR}} \mid X}  = \E{\psi_{s=0}^{\tn{CDR}} \mid X} = \tn{CATE} (X,0),
    \]
    where \tn{CATE($X,s$)} is defined in \eqref{eq:cate}, and therefore
    \begin{equation} \label{eq:cdrSig_int}
         \mathbb{E} [\psi^{\tn{CDR}} \mid X] = 0,
    \end{equation}
     if $~\forall s,a \in \cbrc{0,1}$, either $\bar{F}_{s,a} (t \lvert X)$ \eqref{eq:fbar}, or both $\bar{G}_{s,a} (t \lvert X)$ \eqref{eq:gbar} and $P(A=a \lvert X, S=s)$ are correctly estimated.
\end{restatable}  

If \eqref{eq:cdrSig_int} fails to hold, it means that a subset of the  Assumptions~\ref{asm:iva}–\ref{asm:ooo} is violated. It remains to convert this condition into a hypothesis that can be tested using the maximum moment restriction via the machinery of reproducing kernel Hilbert spaces (RKHS) \citep{muandet2020kernel}. Before writing the full characterization of the test, we note that weaker assumptions on external validity also suffice to construct the hypothesis in \eqref{eq:cdrSig_int}.

\begin{restatable}{prop}{propMMR}
    \label{prop:propMMR}
Consider the same setup in \Cref{thm:prethmMMR}, with the only difference being that we assume that
\[
\tn{CATE} (X,0) = \tn{CATE} (X,1),
\]
instead of the stronger \enquote{ignorability of selection} in \Cref{asm:evo}, where \tn{CATE($X,s$)} is defined in \eqref{eq:cate}. Then, results in \Cref{thm:prethmMMR} continue to hold.
\end{restatable}  

\Cref{prop:propMMR} highlights that we are effectively testing a weaker assumption. This is desirable when one is not necessarily interested in whether the conditional potential outcomes have identical distributions in both studies, but only if the CATE functions are the same. The SPRINT trial~\citep{sprint2015randomized} is an example where only the latter holds.

\begin{restatable}[MMR-based test for validity assumptions with conditionally independent censoring]{thm}{thmMMR}
    \label{thm:thmMMR}
    Let $\psi = \psi^{\tn{CDR}}$ and suppose that $P(S=1 \lvert X)$ is correctly estimated. Suppose that either $\bar{F}_{s,a} (t \lvert X)$ \eqref{eq:fbar}, or {\em both} $\bar{G}_{s,a} (t \lvert X)$ \eqref{eq:gbar} and $P(A=a \lvert X, S=s)$ are correctly estimated $\forall s,a \in \cbrc{0,1}$. Let $k(\cdot,\cdot)$ be an ISPD\footnote{$k(\cdot,\cdot): {\cal X} \times {\cal X} \rightarrow \mathbb{R}$ is said to be integrally strictly positive definite (ISPD) if for all $f: {\cal X} \rightarrow \mathbb{R}$ satisfying $0< \|f\|_2^2 < \infty$, we have $\int_{{\cal X} \times {\cal X}} f(x)k(x, x')f(x') \dif x \dif x' > 0$.}, continuous, and bounded kernel, and ${\cal F}$ be the RKHS endowed with $k(\cdot,\cdot)$. Suppose that $\lvert \E{\psi\lvert X} \rvert  < \infty$ and $\mathbb{E} \big[[\psi k(X, X') \psi']^2\big] < \infty$ a.s. in $P_X$, where $(\psi, X)$ and $(\psi', X')$ are i.i.d. Let $\mathbb{M} = \sup_{f \in {\cal F}, ||f|| \le 1}\left(\E{\psi f(X)}\right)^2$ be the maximum moment restriction (MMR). Then, under Assumptions~\ref{asm:iva}–\ref{asm:ooo}, the conditional moment restriction $\E{\psi \lvert X} = 0$ holds $P_X$ a.s., which implies that the following null hypothesis $H_0$ holds.
    \[
        H_0:~\mathbb{M}^2 = 0, \qquad \qquad H_1:~\mathbb{M}^2 \neq 0.
    \]
    We can then use the following empirical estimate of $\mathbb{M}^2$ as the test statistic,
    \begin{equation}
    \label{eq:test_statistic}
        \hat{\mathbb{M}}_n^2 = \frac{1}{n(n-1)}\sum_{i,j \in {\cal I}, i \ne j} \psi_i k(x_i, x_j)\psi_j.
    \end{equation}
    which has the following asymptotic distributions under the null $H_0$ and the alternative $H_1$ hypotheses.
    \begin{align*}
        \tn{Under}~H_0&:~\hat{\mathbb{M}}_n^2 \xrightarrow[]{d} \sum_{j=1}^\infty \lambda_j(Z_j^2 - 1). \\ \tn{Under}~H_1&:~\sqrt{n}(\hat{\mathbb{M}}_n^2 - \mathbb{M}^2) \xrightarrow[]{d} {\cal N}(0, 4\sigma^2).
    \end{align*}
     where $Z_j$ are i.i.d. standard normal variables and $\lambda_j$ are the eigenvalues of $\psi k(x,x') \psi'$, and $\sigma^2= \tn{Var}_{(\psi, X)} \left(\mathbb{E}_{(\psi', X')}[\psi k(X,X') \psi']\right)$.
\end{restatable}

\Cref{thm:thmMMR} formalizes the implications of Assumptions \ref{asm:iva}-\ref{asm:ooo} as a null hypothesis we can test by calculating a statistic (see \eqref{eq:test_statistic}) from the data and compare against a threshold $t_{\alpha}$ where $\alpha \in (0,1)$ controls the acceptable \enquote{risk} of falsifying the assumptions when they are indeed true. We provide the explicit steps in \Cref{alg:high_level}.
\begin{algorithm}[t]
   \caption{Testing for internal and external validity under conditionally independent right-censoring}\label{alg:high_level}
\begin{algorithmic}
   \STATE {\bfseries Input:} Combined sample from RCT $S=0$ and OS $S=1$: $\{X_i, A_i, S_i, \tilde{Y}_i, \Delta_i \}_{i=1}^{n}$, desired test level $\alpha$
   \STATE {\textbf{1.} Estimate $P(S|X),~P(A|X,S),~\bar{G},~\bar{F}$}
   \STATE {\textbf{2.} Compute $\psi^{\tn{CDR}}_i$ in \eqref{eq:cdrSig} for $i \in \cbrc{1,\ldots,n}$}
   \STATE {\textbf{3.} Compute test statistic $\mathbb{M}^2$ in \eqref{eq:test_statistic}} for $\psi_i\!=\psi^{\tn{CDR}}_i$
   \STATE {\textbf{4.} Compute test threshold $t_{\alpha}$} (see~\Cref{app:test_details})
   \STATE {\textbf{5.} \textbf{if} $t_{\alpha} < \alpha$ \textbf{then} reject $H_0$ \textbf{else} accept $H_0$}
\end{algorithmic}
\end{algorithm}

Note that the censoring assumptions cannot be verified separately. As such, linking the rejection of the test to the violation of the validity assumptions is difficult. Nevertheless, there are cases where the censoring assumptions are true by design. Consider a study where patients are recruited at different times and followed until a fixed endpoint. In that case, the censoring time is known as soon as a patient enters the study (type-1 censoring \citep{leung1997censoring}), and is independent from their time-to-event outcome.
\subsection{Falsification with Global Censoring} \label{ssec:fwgc}
While conditionally independent censoring may be plausible, it is easy to imagine settings where it is not. This makes it challenging to attribute the rejection of the test the violation of the validity assumptions rather than censoring assumptions. For instance, consider a study where patients are more likely to drop out after experiencing adverse side effects, and a short censoring time is associated with a short survival time. This induces {\em dependent} censoring and renders \Cref{asm:cic} implausible. Therefore, it is critical to understand how a benchmarking procedure fares under dependent censoring models \citep{gharari2023copula}. 

In this section, we introduce an alternative (and perhaps more plausible) censoring mechanism, which we refer to as {\em global censoring}, and show that the $\psi^{\tn{CDR}}$ signal in \eqref{eq:cdrSig} can still be used to test the validity assumptions. Global censoring allows for dependent censoring, contingent on the conditional distribution of the censoring time $C$ being identical in the RCT and the OS, reflecting the intuition that the censoring mechanism is the same in RCT and OS populations. 

\begin{restatable}[\textit{Global censoring}]{asm}{cic}
    \label{asm:gc}
    \[
        C \indep S \mid Y, X, A.
    \]
\end{restatable}

\paragraph{CDR Signal with Global Censoring} 
Global censoring does not entail conditionally independent censoring. Therefore, the CATE is not necessarily identifiable in the RCT or the OS. This challenges the core idea in \Cref{ssec:fwic} where internal validity and conditionally independent censoring imply that the instance-wise signals $\psi_{s=0}^{\tn{CDR}}$ and $\psi_{s=1}^{\tn{CDR}}$ in \eqref{eq:cdrSig} are unbiased for the CATE$(X,0)$ and CATE$(X,1)$ in \eqref{eq:cate}. Since CATE$(X,0)$= CATE$(X,1)$ by external validity, we proposed testing the equivalence of two signals as a proxy for testing the validity assumptions (see \Cref{thm:thmMMR}).

Nevertheless, we can show that the global censoring assumptions also imply the equivalence of $\psi_{s=0}^{\tn{CDR}}$ and $\psi_{s=1}^{\tn{CDR}}$, even if these signals are no longer unbiased for the CATE anymore. Crucially, this means that the same falsification test in \Cref{thm:thmMMR} can also be used under global censoring, as we show next.

\begin{restatable}[]{thm}{gcThm}
    \label{thm:gcThm}
    If Assumptions \ref{asm:iva},\ref{asm:evo},\ref{asm:gc} hold, we have
    \begin{equation}
        \mathbb{E} [\psi^{\tn{CDR}} \mid X] = \mathbb{E} [\psi^{\tn{CDR}}_{s=1} - \psi^{\tn{CDR}}_{s=0} \mid X]= 0, \label{eq:gccdrthm}
    \end{equation} 
    where $\psi^{\tn{CDR}}_{s}$ is defined in (\ref{eq:dr_signal},\ref{eq:cdrSig}).
\end{restatable}

\begin{proof} [Proof Sketch]
Even though $\psi^{\tn{CDR}}_{s}$ are biased for the CATE$(X,s)$ in general (as opposed to \Cref{thm:prethmMMR}), the bias in the RCT $S=0$ and the OS $S=1$ will be the same due to \Cref{asm:gc}; therefore the conditional moment restrictions in \eqref{eq:gccdrthm} will still hold.
\end{proof}
\Cref{thm:gcThm} allows to use $\psi^{\tn{CDR}}$ to under global censoring through the same machinery in \Cref{thm:thmMMR}. This property is very desirable as it makes testing the validity assumptions, which is our original motivation, more tangible by providing a falsification test that works under two different censoring mechanisms, significantly increasing the generality of the procedure.

\paragraph{Alternative Signals for Global Censoring}
The global censoring assumption also allows the construction of more straightforward signals to test the validity assumptions. For instance, we propose the following IPW signals that use $\tilde{Y}$ in \eqref{eq:ytilde_def}:
\begin{align}
    \psi^{\tn{IPW},\tilde{Y}}_{s} &\coloneqq \frac{\ind{S=s,A=1} \tilde{Y}}{P(S=s,A=1 \lvert X)} - \frac{\ind{S=s,A=0} \tilde{Y}}{P(S=s,A=0 \lvert X)}. \nonumber \\
    \psi^{\tn{IPW},\tilde{Y}} &\coloneqq \psi^{\tn{IPW},\tilde{Y}}_{s=1} - \psi^{\tn{IPW},\tilde{Y}}_{s=0}. \label{eq:gcipw}
\end{align}

\begin{restatable}[]{thm}{gcThmtwo}
    \label{thm:gcThm2}
    If Assumptions \ref{asm:iva},\ref{asm:evo},\ref{asm:gc} hold, we have
    \begin{equation}
        \mathbb{E} [\psi^{\tn{IPW},\tilde{Y}} \mid X] = \mathbb{E} [\psi^{\tn{IPW},\tilde{Y}}_{s=1} - \psi^{\tn{IPW},\tilde{Y}}_{s=0} \mid X] = 0. \label{eq:gcipwthm}
    \end{equation} 
\end{restatable}

$\psi^{\tn{IPW},\tilde{Y}}_s$ imputes censored outcomes {\em directly} with the censoring time. One could use another IPW signal after {\em dropping} the censored data, and also doubly-robust signals with additional mean outcome estimators ({\em e.g.}, for $\tilde{Y}$). We provide the details for alternative signals, which will also serve as \enquote{baselines} in the experiments, in \Cref{app:baseline_signals}. Note that even though we call them baselines for their naive approach, these signals have not been considered for this problem before. 

The weaker version of the external validity assumption {\em cannot} be tested under global censoring, whereas this was possible with conditionally independent censoring (see \Cref{prop:propMMR}). Our next result formalizes this, where we consider an even stronger alternative than the exchangeability of CATEs in \Cref{prop:propMMR}.  

\begin{restatable}[]{prop}{propNoExMepo}
   \label{prop:NoExMepo}
    Consider the same setup in Theorems \ref{thm:gcThm} and \ref{thm:gcThm2}, with the only difference being we assume
    \[
        \E{Y(a) \mid X,S=0} = \E{Y(a) \mid X,S=1},
    \]
    $\forall a \in \cbrc{0,1}$, instead of the stronger \enquote{ignorability of selection} in \Cref{asm:evo}. Then, \eqref{eq:gccdrthm} and \eqref{eq:gcipwthm} are not true in general.
\end{restatable}

\section{IHDP EXPERIMENTS} \label{sec:semiexp}
\begin{table*}[ht] 

\centering

\caption{{\em Semi-synthetic IHDP experiments with conditionally independent censoring}. Entries are the rejection rates of the null hypothesis (see \Cref{thm:thmMMR}) over 40 independent runs. $\lvert \beta_{\mathrm{prop}} \lvert$ quantifies the severity of unmeasured confounding, {\em i.e.}, internal validity violation (\hyperref[asm:iva]{A1}), and $\Delta \beta_{\mathrm{Cox}}$ the severity of external validity violation (\hyperref[asm:evo]{A2}).}

\label{table:cic_synth_exp}

\begin{tabular}{lccccccccccc}

\toprule
\toprule

\textbf{Setup \#} & \multicolumn{2}{c}{\textbf{1}} & \multicolumn{2}{c}{\textbf{2}}  & \multicolumn{2}{c}{\textbf{3}}  & \multicolumn{2}{c}{\textbf{4}}  & \multicolumn{2}{c}{\textbf{5}} \\

\cmidrule(lr){1-1} \cmidrule(lr){2-3} \cmidrule(lr){4-5} \cmidrule(lr){6-7} \cmidrule(lr){8-9} \cmidrule(lr){10-11}

\textbf{Assumption validity}
& 
\multicolumn{2}{c}{\textbf{\mgray{\hyperref[asm:iva]{A1}}} \cblue{\checkmark} \textbf{\mgray{\hyperref[asm:evo]{A2}}} \cblue{\checkmark}}  & 
\multicolumn{2}{c}{\textbf{\mgray{\hyperref[asm:iva]{A1}}} \cblue{\checkmark} \textbf{\mgray{\hyperref[asm:evo]{A2}}} \cred{\ding{55}} } & 
\multicolumn{2}{c}{\textbf{\mgray{\hyperref[asm:iva]{A1}}} \cblue{\checkmark} \textbf{\mgray{\hyperref[asm:evo]{A2}}} \cred{\ding{55}} } &
\multicolumn{2}{c}{\textbf{\mgray{\hyperref[asm:iva]{A1}}} \cred{\ding{55}} \textbf{\mgray{\hyperref[asm:evo]{A2}}}  \cblue{\checkmark}} & 
\multicolumn{2}{c}{\textbf{\mgray{\hyperref[asm:iva]{A1}}} \cred{\ding{55}} \textbf{\mgray{\hyperref[asm:evo]{A2}}} \cblue{\checkmark}}  
\\

\textbf{Violation severity} 
& \multicolumn{2}{c}{---} & \multicolumn{2}{c}{$\Delta \beta_{\mathrm{Cox}} = 0.2$}  & \multicolumn{2}{c}{$\Delta \beta_{\mathrm{Cox}} = 1$} & \multicolumn{2}{c}{$\lvert \beta_{\mathrm{prop}} \lvert = 1$}  & \multicolumn{2}{c}{$\lvert \beta_{\mathrm{prop}} \lvert = 2.5$}   \\

\textbf{Metric}
& \multicolumn{2}{c}{\textcolor{red}{\textbf{Type-1 error}}} & \multicolumn{2}{c}{\cblue{\textbf{Power}}}  & \multicolumn{2}{c}{\cblue{\textbf{Power}}}  & \multicolumn{2}{c}{\cblue{\textbf{Power}}}  & \multicolumn{2}{c}{\cblue{\textbf{Power}}}  \\

\cmidrule(lr){1-1} \cmidrule(lr){2-3} \cmidrule(lr){4-5} \cmidrule(lr){6-7} \cmidrule(lr){8-9} \cmidrule(lr){10-11}

\textbf{OS sample size}, $n_1$ & \textbf{985} & \textbf{2955} & \textbf{985} & \textbf{2955} & \textbf{985} & \textbf{2955} & \textbf{985} & \textbf{2955} & \textbf{985} & \textbf{2955} & \\

\cmidrule(lr){1-1} \cmidrule(lr){2-3} \cmidrule(lr){4-5} \cmidrule(lr){6-7} \cmidrule(lr){8-9} \cmidrule(lr){10-11}
$\tn{DR}-\tilde{Y}$ & 1 & 1 & 0.35 & 0.375 & 1 & 1 & 1 & 1 & 1 & 1 \\
$\tn{DR}-Y$ & 1 & 1 & 0.55 & 0.6 & 0.85 & 0.95 & 1 & 1 & 1 & 1 \\
$\tn{IPW}-\tilde{Y}$ & 1 & 1 & 0 & 0 & 0.125 & 0.35 & 1 & 1 & 1 & 1 \\
$\tn{IPW}-Y$ & 0.9 & 1 & 0 & 0 & 0 & 0.025 & 1 & 1 & 1 & 1 \\

\textbf{IPCW} & 0 & 0 & 0 & 0 & 0.425 & 0.925 & 0.025 & 0.05 & 0.825 & 0.95 \\
\textbf{CDR} & 0 & 0.025 & 0.2 & 0.3 & 0.9 & 0.975 & 0.275 & 0.425 & 0.8 & 0.85 \\

\bottomrule
\bottomrule

\end{tabular}
\end{table*}
The Infant Health and Development Program (IHDP) is an RCT that studied the effect of professional home visits on cognitive abilities in premature infants, with a sample size of 985 \citep{brooks1992effects}. We use covariate information from the IHDP trial to form semi-synthetic OS and RCT cohorts. We then simulate the binary treatment assignments, time-to-event outcomes, and censoring times based on the covariate information, as detailed in \Cref{ssec:dgp}. The simulations cover settings with different validity assumption violations and censoring mechanisms: conditionally independent and global. Using the simulated data, we compute various signals, including the $\psi^{\tn{CDR}}$ signal in \eqref{eq:cdrSig} and some rudimentary alternatives to serve as baselines, which are described in \Cref{ssec:abl}. We then conduct falsification tests for the validity assumptions using different signals as $\psi$ signal in the 3rd step of \Cref{alg:high_level} and compare the type-1 errors and powers. Our code is available at \url{https://github.com/demireal/censored-mmr}.
\subsection{Data-Generating Process} \label{ssec:dgp}
For a patient with covariates $X_i$ in study $S_i$, we sample a binary treatment $A_i \sim \texttt{Bernoulli} \left( P(A=1 \lvert X_i,S_i) \right)$. The propensity score $P(A \lvert X, S)$ is set to a sigmoid function for both studies. Then we sample time-to-event $Y_i$ and censoring time $C_i$ outcomes according to the survival functions $\bar{F}_{S_i,A_i}(t \lvert X_i)$ and $\bar{G}_{S_i,A_i}(t \lvert X_i)$. We adopt a CoxPH framework to model the effect of covariates $X$ and specify $\bar{F}_{S,A}(t \lvert X)$ and $\bar{G}_{S,A}(t \lvert X)$ as
\begin{equation} \label{eq:surv_func_exp}
    \overline{W}_0 (t; \lambda, p)^{\exp(X^{\intercal} \beta_{\mathrm{Cox}})},
\end{equation}
where $\overline{W}_0 (t; \lambda, p)$ is the {\em Weibull} baseline survival function, and the parameters $\lambda$, $p$, and $\beta$ can be set differently for each study $S$ and treatment group $A$. Exact expressions and specific parametrizations used in the experiments can be found in \Cref{app:ihdp}.
\subsection{Ablation Studies} \label{ssec:abl}
We perform ablation studies to measure the efficiency of our CDR signal in \eqref{eq:cdrSig} for testing the validity assumptions. For comparison, we propose \enquote{baselines} with no component to model censoring: $\tn{IPW}-Y$, $\tn{DR}-Y$, $\tn{IPW}-\tilde{Y}$, $\tn{DR}-\tilde{Y}$ (detailed in \Cref{app:baseline_signals}). $\tn{IPW}-Y$ and $\tn{DR}-Y$ drop the censored data.  $\tn{IPW}-\tilde{Y}$ and $\tn{DR}-\tilde{Y}$ impute the missing time-to-event with censoring time. For instance, the $\tn{IPW}-\tilde{Y}$ baseline uses $\psi = \psi^{\tn{IPW},\tilde{Y}}$ signal \eqref{eq:gcipw} in the 3rd step of \Cref{alg:high_level}. DR baselines employ additional estimators for imputed or uncensored outcomes. We also adopt an inverse propensity of censoring-weighted (IPCW) signal that accounts for censoring by inverse-weighting with $\bar{G}$. 

The significance level of the tests is set to $\alpha = 0.05$ as the desired type-1 error threshold. The synthetic RCT cohort size is the original IHDP cohort size $n_0 = 985$. We experiment with two OS cohort sizes $n_1 = 985$ and $n_1 = 2955$, where we copy the covariate data from the IHDP three times for the former. Even though the covariates are repeated, treatment and time-to-event generation processes still involve stochasticity. 
\subsection{Conditionally Independent Censoring} \label{sssec:cic_exp}
We follow the data-generating process in \Cref{ssec:dgp}, ensuring that censoring Assumptions~\ref{asm:cic} and \ref{asm:ooo} hold and consider various violations of validity assumptions. The results are presented in \Cref{table:cic_synth_exp} and \Cref{fig:cicsweep}. 

We start with setup \#1, where validity assumptions hold. The baselines have very high type-1 errors, falsifying the OS despite being compatible with the RCT. Increasing the sample size does not alleviate the problem, as baselines' CATE estimates are asymptotically biased. The IPCW and CDR signals boast significantly lower type-1 errors, maintaining the test level of 0.05.

Next, we consider the violation of the external validity assumption \mgray{\hyperref[asm:evo]{A2}}; where one of the $\beta_{\mathrm{Cox}}$ parameters in \eqref{eq:surv_func_exp} is different between the RCT and OS (setups \#2 and \#3 in \Cref{table:cic_synth_exp}, and top left in \Cref{fig:cicsweep}). By $\Delta \beta_{\mathrm{Cox}}$, we denote the magnitude of the difference, where a larger value causes a more severe violation. In addition to having high type-1 errors, IPW-based baselines also suffer from low power. This behavior can be expected, {\em e.g.}, if the bias introduced by naively handling the censored data cancels part of the bias from violating external validity or due to the high variance of IPW-based estimators. IPCW reacts to more severe violations of external validity; however, it cannot detect milder violations. The CDR signal enjoys higher power as a meaningful complement to its low type-1 error. 

We then consider the violation of internal validity \mgray{\hyperref[asm:iva]{A1}} in the OS by introducing \enquote{unobserved confounding} (UC) (setups \#4 and \#5 in \Cref{table:cic_synth_exp} and top middle in \Cref{fig:cicsweep}). We conceal the confounding covariate \enquote{sex} and adjust the violation severity through its effect on the propensity score of treatment, captured by $\beta_{\mathrm{prop}}$. We observe that the CDR signal can detect unobserved confounding, an ability pronounced by increased sample size and violation severity. As before, IPCW does not react when the violation is subtle.

Overall, falsification with CDR signal has the most reliable performance with low type-1 error and ability to detect milder violations. We also verify its doubly-robust property in \Cref{supssec:drtest}. Further, in \Cref{supssec:witfunc}, we show that a \enquote{witness function} may provide explanations by revealing the regions of $\mathcal{X}$ where CATE estimates from RCT and OS differ the most.
\subsection{Global Censoring} \label{sssec:global_cen}
To simulate the global censoring mechanism, we censor patients whose time-to-event outcome exceeds a threshold by setting the censoring time to a smaller value than the threshold through the same mechanism in RCT and OS. We present the type-1 errors and powers at varying levels of validity violations in \Cref{fig:cicsweep}. In contrast to the conditionally independent censoring case, all signals maintain low type-1 errors, corroborating the theory of \Cref{ssec:fwgc}. IPW and IPCW signals have lower type-1 errors than their doubly-robust counterparts; however, they are not as well-powered to detect violations of the validity assumptions.
\begin{figure}[t]
\centering
\includegraphics[width=.48\textwidth]{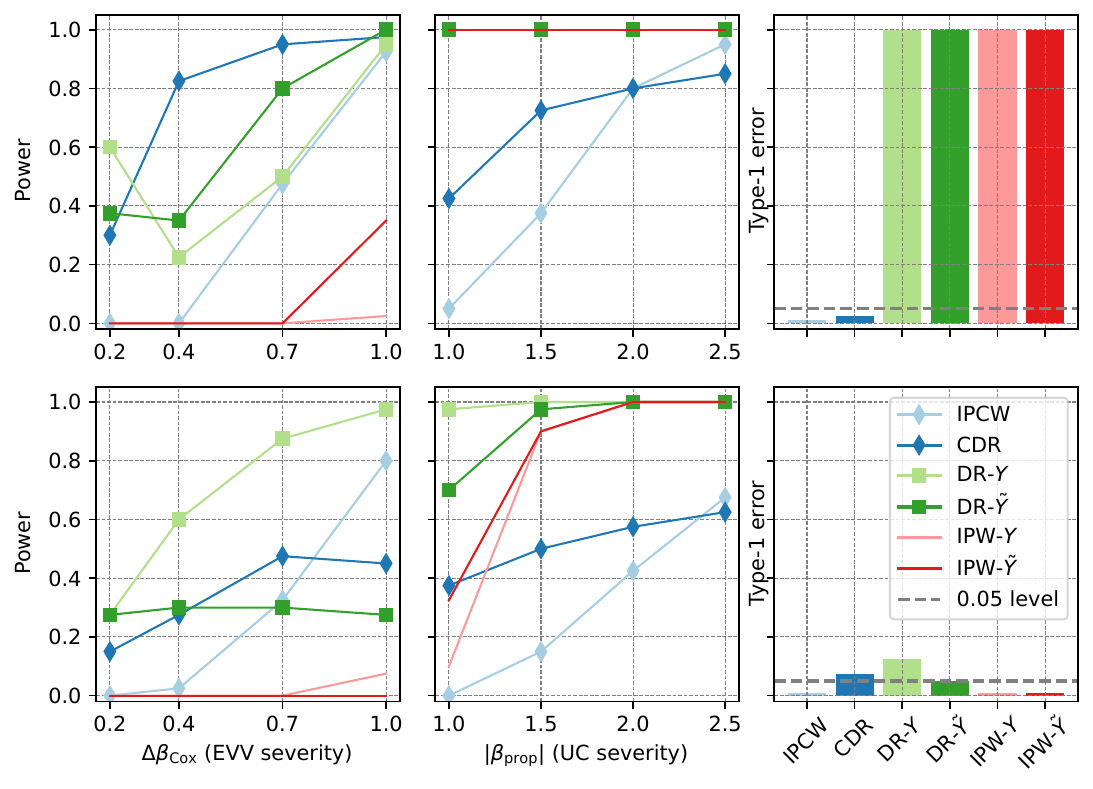}
\caption{{\em Top row:} \mgray{\hyperref[asm:cic]{Conditionally independent censoring}} results. {\em Bottom row:} \mgray{\hyperref[asm:gc]{Global censoring}} results. EVV = External validity violation (\mgray{\hyperref[asm:evo]{A2}}). UC = Unobserved confounding (\mgray{\hyperref[asm:iva]{A1}}). OS size is $n_1 = 2955$.}
\label{fig:cicsweep}
\rule{\linewidth}{.75pt} 
\end{figure}
\section{WHI EXPERIMENTS} \label{sec:whi}
The Women's Health Initiative (WHI) was launched in the early 1990s to study various health outcomes in postmenopausal women. Previous studies noted discrepancies between the RCT and OS components of the WHI \citep{prentice2005combined}. We focus on the effect of combination hormone therapy on a composite outcome: the minimum time-to-event among multiple endpoints such as heart failure, cancer, and death. 

Similar to \citet{hussain2023falsification}, we limit the follow-up to seven years since the treatment assignment. While previous studies discarded censoring by binarizing the time-to-event outcome and imputing the outcome for censored patients with $Y=0$, our framework allows us to use the true outcomes and explicitly model the censoring component. We used 954 features available in the RCT and OS and performed principal component (PC) analysis to alleviate collinearity-related instabilities in Cox regression models (350 PCs, capturing 90\% of the variance, details in~\Cref{app:whi}). We split the data into ten folds, estimate the nuisance functions using nine folds (step 1 in \Cref{alg:high_level}), and perform the MMR test with the remaining fold. We repeat ten times and report the average rejection rates. 

\begin{table}[ht]
  \centering
  \caption{WHI experiments. Average rejection rate over 10-folds for different fractions of selection biases $f$.}
  \begin{tabular}{lcccc}
    \toprule
    \toprule
               & $\tn{IPW}-\tilde{Y}$ & $\tn{DR}-\tilde{Y}$  & IPCW & CDR   \\
    \midrule
     $f=0$   & 0   & 0.5   & 0 & 0.6 \\
     $f=0.1$   & 0.1   & 0.2   & 0.1 & 0.9  \\
     $f=0.25$   & 0   & 0.4   & 0.3 & 0.9  \\
    \bottomrule
    \bottomrule
  \end{tabular}
  \label{tab:whi}
\end{table}

We also simulate selection bias by dropping $f \in (0,1)$ fraction of patients from the RCT's control group that experienced an event. The results are presented in \Cref{tab:whi}. IPW-based methods have lower rejection rates, while the CDR signal has the highest rejection rate, which increases with the selection bias $f$.
\section{RELATED WORK IN EPIDEMIOLOGY}
Evaluating the internal and external validity of RCTs and OSes has been of significant interest in the epidemiology literature. Here, we surface some references for the interested reader.

One elegant concept is {\em negative outcomes} \citep{lipsitch2010negative, sofer2016negative}. A negative outcome is known or expected to be unaffected by the treatment. Therefore, a significant difference in a negative outcome between the treatment and control groups may point to flaws in the study design. For instance, \citet{dagan2021bnt162b2} leverage the fact that the COVID-19 vaccine should not have a significant effect within the first few days of administration to guide their adjustment for confounders in the observational data.

\citet{viele2014use} investigate methods to incorporate {\em historical controls} into the design of new RCTs to make them more efficient. Their \enquote{test-then-pool} procedure first compares the historical controls to trial controls before pooling them. \citet{Hartman2015-td} study estimating the {\em population} average treatment effect {\em on the treated} from an RCT, where the population of interest is defined by an OS cohort. Their approach involves a placebo test in the first step to gauge the generalizability of the trial to the target population. \citet{de2014testing} show how an alternative set of causal assumptions on {\em instrumental variables} can be used to construct a test for the no unmeasured confounding assumption in an OS.

We close by noting \citet{forbes2020benchmarking} and \citet{wang2023emulation}. They provide thorough empirical evidence regarding the compatibility of RCTs and their OS counterparts by analyzing complementary findings from an extensive set of studies in the literature.
\section{CONCLUSION}
We developed a framework to test the validity of an OS by benchmarking it against an RCT, when the outcomes are right-censored. We considered the common conditionally independent censoring condition and introduced a novel one: global censoring. We demonstrated that naively handling the censoring leads to unreliable tests. In contrast, our censoring-doubly-robust signal facilitated tests with low type-1 error and high sensitivity to violations of the validity assumptions under both censoring scenarios, making it a promising candidate for generally applicable benchmarking procedures under different censoring scenarios.
\subsubsection*{Acknowledgments}
The authors thank the anonymous reviewers for their helpful suggestions and Ming-Chieh Shih for discussions during the earlier versions of the manuscript. ID was supported by funding from the Eric and Wendy Schmidt Center at the Broad Institute of MIT and Harvard. ZH was supported by an ASPIRE award from The Mark Foundation for Cancer Research and by the National Cancer Institute of the National Institutes of Health under Award Number F30CA268631. The content is solely the responsibility of the authors and does not necessarily represent the official views of the National Institutes of Health. MO and DS were supported in part by Office of Naval Research Award No. N00014-21- 1-2807. EDB was funded by a FWO-SB PhD grant. This manuscript was prepared using WHI-CTOS Research Materials obtained from the National Heart, Lung, and Blood Institute (NHLBI) Biologic Specimen and Data Repository Information Coordinating Center and does not necessarily reflect the opinions or views of the WHI-CTOS or the NHLBI.
\bibliography{ref}
\bibliographystyle{plainnat}

\section*{Checklist}


 \begin{enumerate}

 \item For all models and algorithms presented, check if you include:
 \begin{enumerate}
   \item A clear description of the mathematical setting, assumptions, algorithm, and/or model. [Yes. See Sections~\ref{sec:notset} and \ref{sec:prelim}.]
   \item An analysis of the properties and complexity (time, space, sample size) of any algorithm. [Yes. See Section~\ref{sec:prelim}.]
   \item (Optional) Anonymized source code, with specification of all dependencies, including external libraries. [Yes. Link to the code repository is provided in \Cref{sec:semiexp}.]
 \end{enumerate}

 \item For any theoretical claim, check if you include:
 \begin{enumerate}
   \item Statements of the full set of assumptions of all theoretical results. [Yes. All the assumptions and theoretical results are listed in the main paper in Sections~\ref{sec:notset} and \ref{sec:prelim}.]
   \item Complete proofs of all theoretical results. [Yes. All the proofs are included in \Cref{sec:proofs}.]
   \item Clear explanations of any assumptions. [Yes. The motivations/insights behind the assumptions are listed either before or after stating the critical assumptions.]     
 \end{enumerate}

 \item For all figures and tables that present empirical results, check if you include:
 \begin{enumerate}
   \item The code, data, and instructions needed to reproduce the main experimental results (either in the supplemental material or as a URL). [Yes. See the repository link in \Cref{sec:semiexp}]
   \item All the training details ({\em e.g.}, data splits, hyperparameters, how they were chosen). [Yes. See Sections~\ref{sec:semiexp}, \ref{sec:whi} and \Cref{app:datasets}]
         \item A clear definition of the specific measure or statistics and error bars ({\em e.g.}, with respect to the random seed after running experiments multiple times). [Yes]
         \item A description of the computing infrastructure used. ({\em e.g.}, type of GPUs, internal cluster, or cloud provider). [Yes]
 \end{enumerate}

 \item If you are using existing assets ({\em e.g.}, code, data, models) or curating/releasing new assets, check if you include:
 \begin{enumerate}
   \item Citations of the creator If your work uses existing assets. [Not Applicable]
   \item The license information of the assets, if applicable. [Not Applicable]
   \item New assets either in the supplemental material or as a URL, if applicable. [See the repository link in \Cref{sec:semiexp}]
   \item Information about consent from data providers/curators. [Yes. See \Cref{app:whi}]
   \item Discussion of sensible content if applicable, {\em e.g.}, personally identifiable information or offensive content. [Yes. See \Cref{app:whi}]
 \end{enumerate}

 \item If you used crowdsourcing or conducted research with human subjects, check if you include:
 \begin{enumerate}
   \item The full text of instructions given to participants and screenshots. [Not Applicable]
   \item Descriptions of potential participant risks, with links to Institutional Review Board (IRB) approvals if applicable. [Not Applicable]
   \item The estimated hourly wage paid to participants and the total amount spent on participant compensation. [Not Applicable]
 \end{enumerate}

 \end{enumerate}

\clearpage
\onecolumn
\appendix

\addcontentsline{toc}{section}{Appendix} 
\part{Appendix} 
\parttoc 
\newtheorem{suplem}{Lemma}[section]
\section{PROOFS} \label{sec:proofs}

\begin{restatable}[]{suplem}{suplemCateId}
    \label{suplem:CateId}
    Suppose that Assumption \ref{asm:iva} holds. We have, $\forall s, a \in \cbrc{0,1}$
    \begin{equation*}
        P(Y(a) \lvert X, S=s) = P(Y \lvert X, S=s, A=a)
    \end{equation*}
\end{restatable}

\begin{proof}
    \begin{align*}
    P(Y(a) \lvert X, S=s) &= P(Y(a) \lvert X, A=a, S=s) \\
    &= P(Y \lvert X, A=a, S=s)
\end{align*}
by no unobserved confounding and consistency.
\end{proof}
\subsection{Lemma 1}
\lemDR*
\begin{proof}
    \begin{align*}
        \E{\psi^*_{s,a} \lvert X,S=s,A=a} &= \E{Y \lvert X,S=s,A=a} \tag{Theorem 1 in \cite{rubin2007doubly}} \\ 
        &= \E{Y(a) \lvert X,S=s} \tag{Lemma~\ref{suplem:CateId}}
    \end{align*} 
\end{proof}
\subsection{Corollary 1}
\corDR*
\begin{proof}
    Let us denote by $\hat{P}(A=a \lvert X, S=s)$ the estimated propensity score.
    \begin{align*}
        &  \E{\frac{\ind{S=s,A=a} \psi^*_{s,a}}{P(S=s \lvert X) \hat{P}(A=a \lvert X, S=s)} \Big\lvert X}\\
        &= \E{\frac{ \psi^*_{s,a}}{P(S=s \lvert X) \hat{P}(A=a \lvert X, S=s)} \Big\lvert X, S=s, A=a} P(S=s, A=a \lvert X) \\
        &= \E{\psi^*_{s,a} \lvert X, S=s, A=a} \\
        &= \E{Y(a) \lvert X,S=s} \tag{Lemma~\ref{lem:lemDR}}
    \end{align*}
    where the probability terms cancel out since $\hat{P}(A=a \lvert X, S=s) = P(A=a \lvert X, S=s)$ is correctly estimated.
\end{proof}
\subsection{Theorem 1}
\thmDRtwo*
\begin{proof}
Let us denote the estimates for the nuisance functions with
\begin{equation*}
    \hat{P} (A=a \lvert X,S=s),\quad \hat{G}_{s,a} (t \lvert X),\quad \hat{F}_{s,a} (t \lvert X),\quad \hat{\mu}_{s,a} (X)
\end{equation*}
and we have
\begin{align*} 
    \hat{\psi}_{s,a}^{\tn{CDR}} = \frac{\ind{S=s}}{P(S=s \mid X)} \times \left( \frac{\ind{A=a} (\hat{\psi}^{*}_{s, a} - \hat{\mu}_{s, a}(X))}{\hat{P}(A=a \mid X,S=s)} + \hat{\mu}_{s, a}(X) \right)
\end{align*}
First, assume that the survival function of the time-to-event outcome is correctly estimated but the survival function of the censoring time and the propensity score are not. That is,
\begin{align}
    \hat{F}_{s,a} (t \lvert X) &= \bar{F}_{s,a} (t \lvert X) \label{eq:fcorrs0} \\
    \hat{\mu}_{s,a} (X) &= \mu_{s,a} (X) \label{eq:fcorrs} \\
    \hat{G}_{s,a} (t \lvert X) &\neq \bar{G}_{s,a} (t \lvert X) \nonumber \\
    \hat{P} (A=a \lvert X, S=s) &\neq P (A=a \lvert X, S=s) \nonumber
\end{align}
Note that
\begin{align}
    \mu_{s,a} (X) &= \E{Y \lvert X, S=s, A=a} \nonumber \\
    &= \E{Y(a) \lvert X, S=s} \nonumber \\
    &= \E{\hat{\psi}^*_{s,a} \lvert X, S=s, A=a} \label{eq:suppeq2}
\end{align}
where the second equality is by Lemma~\ref{suplem:CateId} and the third by \Cref{lem:lemDR} thanks to \eqref{eq:fcorrs0}. We have
\begin{align}
    \E{\hat{\psi}_{s,a}^{\tn{CDR}} \lvert X} &= \E{\frac{\ind{S=s, A=a} \left(\hat{\psi}^*_{s,a} - \hat{\mu}_{s,a} (X) \right)}{P(S=s \lvert X) \hat{P} (A=a \lvert X,S=s)} \Big\lvert X} + \E{\frac{\ind{S=s} \hat{\mu}_{s,a}(X)}{P(S=s \lvert X)} \Big\lvert X} \nonumber \\
    & = \E{\hat{\psi}^*_{s,a} - \hat{\mu}_{s,a} (X) \lvert X,S=s,A=a} \frac{P(A=a \lvert X,S=s)}{\hat{P}(A=a \lvert X,S=s)} + \hat{\mu}_{s,a}(X) \nonumber \\
    & = \underbrace{\left(\E{\hat{\psi}^*_{s,a} \lvert X,S=s,A=a} -\mu_{s,a}(X) \right)}_{0~\tn{by \eqref{eq:suppeq2}}} \frac{P(A=a \lvert X,S=s)}{\hat{P}(A=a \lvert X,S=s)} + \mu_{s,a}(X) \label{eq:lem4_1} \\
    & = \E{Y(a) \lvert X,S=s} \tag{Lemma~\ref{suplem:CateId}}
\end{align}
where \eqref{eq:lem4_1} follows from \eqref{eq:fcorrs}. Next, assume that the propensity score and the survival function of the censoring time is correctly estimated, but the survival function of the outcome is not,
\begin{align}
    \hat{P} (A=a \lvert X, S=s) &= P (A=a \lvert X, S=s) \label{eq:corspec3} \\ 
    \hat{G} (t \lvert X, S=s, A=a) &= \bar{G} (t \lvert X, S=s, A=a) \label{eq:corspec4} \\ 
    \hat{F} (t \lvert X, S=s, A=a) &\neq \bar{F} (t \lvert X, S=s, A=a) \nonumber
\end{align}
We have
\begin{align}
    \E{\hat{\psi}_{s,a}^{\tn{CDR}} \lvert X} &= \mathbb{E} \bigg[ \frac{\ind{S=s, A=a} \hat{\psi}^*_{s,a}}{P(S=s \lvert X) \hat{P} (A=a \lvert X,S=s)}  - \frac{\ind{S=s, A=a} \hat{\mu}_{s,a} (X)}{P(S=s \lvert X) \hat{P}(A=a \lvert X, S=s)} + \frac{\ind{S=s}\hat{\mu}_{s,a}(X)}{P(S=s \lvert X)} \Big\lvert X  \bigg] \nonumber \\
    &\hspace{-30pt} = \E{\frac{\ind{S=s, A=a} \hat{\psi}^*_{s,a}}{P(S=s,A=a\lvert X)} \Big\lvert X} - \E{\frac{\ind{S=s} \hat{\mu}_{s,a} (X)}{P(S=s \lvert X)} \left( \frac{\ind{A=a} - P (A=a \lvert X,S=s)}{P (A=a \lvert X,S=s)} \right) \Big\lvert X} \label{eq:lem4_3} \\
    &\hspace{-30pt} =\E{\hat{\psi}^*_{s,a} \lvert X,S=s,A=a} - \hat{\mu}_{s,a} (X) \times \underbrace{\E{ \left( \frac{\ind{A=a} - P (A=a \lvert X,S=s)}{P (A=a \lvert X,S=s)} \right) \Big\lvert X, S=s}}_{0} \nonumber  \\
    &\hspace{-30pt} = \E{Y(a) \lvert X,S=s} \label{eq:lem4_4} 
\end{align}
where we have \eqref{eq:lem4_3} by \eqref{eq:corspec3}, and \eqref{eq:lem4_4} follows from Lemma~\ref{lem:lemDR} since we have \eqref{eq:corspec4}.
\end{proof}
\subsection{Theorem 2}
\prethmMMR*
\begin{proof}
Note that $\tn{CATE} (X,0) = \E{Y(1) - Y(0) \lvert X, S=0}$.
\begin{align*}
    \E{\psi_{s=1}^{\tn{CDR}} \lvert X} &= \E{\psi_{s=1,a=1}^{\tn{CDR}} \lvert X} - \E{\psi_{s=1,a=0}^{\tn{CDR}} \lvert X} \nonumber \\
    &= \E{Y(1) \lvert X, S=1} - \E{Y(0) \lvert X, S=1} \tag{\Cref{thm:thmDRtwo}} \\
    &= \E{Y(1) \lvert X, S=0} - \E{Y(0) \lvert X, S=0}  \tag{Ignorability of selection, \Cref{asm:evo}} \\
    &= \E{\psi_{s=0}^{\tn{CDR}} \lvert X} \tag{by symmetry}
\end{align*}
\end{proof}
\subsection{Proposition 1}
\propMMR*
\begin{proof}
    The proof is identical to the proof of \Cref{thm:prethmMMR}, only difference being we invoke the alternative external validity assumption CATE$(X,0)$ = CATE$(X,1)$ in the third step, instead of the stronger ignorability of selection.
\end{proof}
\subsection{Theorem 3}
\thmMMR*
\begin{proof}
     By \Cref{thm:prethmMMR} and \eqref{eq:cdrSig}, we have
     \[
         \E{\psi \lvert X} = \E{\psi^{\tn{CDR}} \lvert X} = 0
     \]
    The hypothesis test results then follow from Theorem 3.1 in \citet{hussain2023falsification}.
\end{proof}

\begin{restatable}[]{suplem}{gclem0}
    \label{suplem:gclem0}
    Suppose that Assumptions \ref{asm:iva}, \ref{asm:evo}, \ref{asm:gc} hold. We have, for all $a \in \cbrc{0,1}$ and $t \in \mathbb{R}_+$
    \begin{align*} 
        \bar{G}_{s=0,a} (t \lvert X) &= \bar{G}_{s=1,a} (t \lvert X) \\ 
        Q_{s=0,a} (X,t) &= Q_{s=1,a} (X,t) \\ 
    \end{align*}
\end{restatable}

\begin{proof}
    Let us start with the first statement.
    \begin{align*}
        \bar{G}_{s=0,a} (t \lvert X) &= P(C > t \lvert X, S=0, A=a) \\
        &= \sum_y P(C > t \lvert Y=y, X, S=0, A=a) P(Y=y \lvert X, S=0, A=a) \\ 
        &= \sum_y P(C > t \lvert Y=y, X, S=0, A=a) P(Y(a)=y \lvert X, S=0) \tag{Lemma~\ref{suplem:CateId}} \\ 
        &= \sum_y P(C > t \lvert Y=y, X, S=1, A=a) P(Y(a)=y \lvert X, S=1) \tag{Assumptions \ref{asm:evo}, \ref{asm:gc}} \\
        &= \bar{G}_{s=1,a} (t \lvert X) \tag{by symmetry}
    \end{align*}
    For the second statement, we write
    \begin{align*}
        Q_{s=0,a} (X,t) &= \E{Y \lvert X, Y>t, S=0,A=a} \\
        &= \sum_y y P(Y=y \lvert X, Y>t, S=0, A=a) \\
        &= \sum_y y \frac{P(Y(a)=y, Y(a)>t \lvert X, S=0)}{P(Y(a)>t \lvert X, S=0)} \tag{Lemma~\ref{suplem:CateId}} \\
        &= \sum_y y \frac{P(Y(a)=y, Y(a)>t \lvert X, S=1)}{P(Y(a)>t \lvert X, S=1)} \tag{Assumption \ref{asm:evo}} \\            
        &= Q_{s=1,a} (X,t) \tag{by symmetry}
    \end{align*}
\end{proof} 

\begin{restatable}[]{suplem}{gclem1}
    \label{suplem:gclem1}
    Suppose that Assumptions \ref{asm:iva}, \ref{asm:evo}, \ref{asm:gc} hold. We have
    \begin{equation*} 
        P(Y\leq C \lvert X,S=0,A) = P(Y \leq C \lvert X,S=1,A) 
    \end{equation*}
\end{restatable}

\begin{proof}
    \begin{align*}
        P(Y\leq C \lvert X,S=0,A) 
        &= \sum_{y} P(y \leq C \lvert Y=y,X,S=0,A) P(Y=y \lvert X,S=0,A) \nonumber \\
        &= \sum_{y} P(y \leq C \lvert Y=y,X,S=0,A) P(Y(A)=y \lvert X,S=0) \tag{Lemma \ref{suplem:CateId}} \\
        &= \sum_{y} P(y \leq C \lvert Y=y,X,S=1,A) P(Y(A)=y \lvert X,S=1) \tag{Assumptions~\ref{asm:evo}, \ref{asm:gc}} \\
        &= P(Y\leq C \lvert X,S=1,A)  \tag{by symmetry}
    \end{align*}
\end{proof}

\begin{restatable}[]{suplem}{gclem2}
    \label{suplem:gclem2}
    Suppose that Assumptions \ref{asm:iva}, \ref{asm:evo}, \ref{asm:gc} hold. We have
    \begin{align*} 
        \E{Y \lvert X, S=0, A, Y\leq C} &= \E{Y \lvert X, S=1, A, Y\leq C} \\
        \E{C \lvert X, S=0, A, Y > C} &= \E{C \lvert X, S=1, A, Y > C}
    \end{align*}
\end{restatable}

\begin{proof}
For the first statement we write
\begin{align}
    \E{Y \lvert X, S=0, A, Y \leq C} &= \sum_y y P(Y=y \lvert X,S=0,A,Y \leq C) \nonumber \\
    &\hspace{-50pt}= \sum_y y \frac{P(Y \leq C \lvert Y=y,X,S=0,A) P(Y=y \lvert X,S=0,A=a)}{P(Y \leq C \lvert X,S=0,A)} \nonumber \\
    &\hspace{-50pt}= \sum_y y \frac{P(y \leq C \lvert Y=y,X,S=0,A) P(Y(A)=y \lvert X,S=0)}{P(Y \leq C \lvert X,S=0,A)} \tag{Lemma~\ref{suplem:CateId}} \\
    &\hspace{-50pt}= \sum_y y \frac{P(y \leq C \lvert Y=y,X,S=1,A) P(Y(A)=y \lvert X,S=1)}{P(Y \leq C \lvert X,S=1,A)} \tag{Lemma~\ref{suplem:gclem1}, Assumptions \ref{asm:evo},\ref{asm:gc}} \\
    &\hspace{-50pt}= \E{Y \lvert X, S=1, A, Y \leq C} \tag{by symmetry} 
\end{align}
For the second statement we write
\begin{align}
    &\E{C \lvert X, S=0, A, Y > C}  \nonumber \\
    =& \sum_y \E{C \lvert X, S=0, A, Y > C, Y=y} P(Y=y \lvert X, S=0, A, Y > C) \nonumber \\
    =& \sum_y \E{C \lvert X, S=0, A, Y > C, Y=y} \frac{P(y > C \lvert Y=y,X,S=0,A) P(Y=y \lvert X,S=0,A)}{P(Y > C \lvert X,S=0,A)} \nonumber \\
    =& \sum_y \E{C \lvert X, S=0, A, Y > C, Y=y} \frac{P(y > C \lvert Y=y,X,S=0,A) P(Y(A)=y \lvert X,S=0)}{P(Y > C \lvert X,S=0,A)} \tag{Lemma~\ref{suplem:CateId}} \\
    =& \sum_y \E{C \lvert X, S=0, A, Y > C, Y=y} \underbrace{\frac{P(y > C \lvert Y=y,X,S=1,A) P(Y(A)=y \lvert X,S=1)}{P(Y > C \lvert X,S=1,A)}}_{\eqqcolon g(X,A,Y=y,S=1)} \tag{Lemma~\ref{suplem:gclem1}, Assumptions \ref{asm:evo}, \ref{asm:gc}} \\
    =& \sum_y g(X,A,Y=y,S=1) \sum_{c:y> c} c P(C=c \lvert X,S=0,A,Y> C,Y=y) \nonumber \\
    =& \sum_y g(X,A,Y=y,S=1) \sum_{c:y> c} c \frac{P(y> C, C=c \lvert X,S=0,A,Y=y)} {P(y> C \lvert X,S=0,A,Y=y)} \nonumber \\
    =& \sum_y g(X,A,Y=y,S=1) \sum_{c:y> c} c \frac{\overbrace{P(y > c \lvert X,S=0,A,C=c,Y=y)}^{=1} P(C=c \lvert X,S=0,A,Y=y)} {P(y > C \lvert X,S=0,A,Y=y)} \nonumber \\
    =& \sum_{y,c: y > c} g(X,A,Y=y,S=1) \frac{c P(C=c \lvert X,S=1,A,Y=y)} {P(y > C \lvert X,S=1,A,Y=y)} \tag{\Cref{asm:gc}} \\
    =&\E{C \lvert X, S=1, A, Y > C} \tag{by symmetry}
\end{align}
\end{proof}
\subsection{Theorem 4}
\gcThm*
\begin{proof}
If we can show that $\E{\psi^{\tn{CDR}}_{s=0,a} \lvert X} = \E{\psi^{\tn{CDR}}_{s=1,a} \lvert X}$ for all $a \in \cbrc{0,1}$, we are done since $\E{\psi^{\tn{CDR}}_{s} \lvert X} = \E{\psi^{\tn{CDR}}_{s,a=1} - \psi^{\tn{CDR}}_{s,a=0} \lvert X}$ for all $s \in \cbrc{0,1}$.

Recall that $\psi_{s,a}^{\tn{CDR}}$ is defined in~\cref{eq:dr_signal} as 
\begin{equation*}
\psi_{s,a}^{\tn{CDR}} \coloneqq \frac{\ind{S=s}}{P(S=s \mid X)} 
\times \left( \frac{\ind{A=a} (\psi^{*}_{s, a} - \mu_{s, a}(X))}{P(A=a \mid X,S=s)} + \mu_{s, a}(X) \right)
\end{equation*}

Allowing us to write
\begin{align}
    &\E{\psi_{s=0,a}^{\tn{CDR}} \lvert X} \nonumber \\
    = &\E{\frac{\ind{S=0, A=a} \psi^*_{s=0,a}}{P (S=0, A=a \lvert X)} \Big\lvert X} - \Cancel{\underbrace{\E{\frac{\ind{S=0, A=a} \mu_{s=0,a} (X)}{P(S=0,A=a \lvert X)} \Big\lvert X}}_{\mu_{s=0,a} (X)}} + \Cancel{\underbrace{\E{\frac{\ind{S=0} \mu_{s=0,a}(X)}{P(S=0 \lvert X)} \Big\lvert X}}_{\mu_{s=0,a} (X)}}  \nonumber \\
    = &\E{\psi^*_{s=0,a} \Big\lvert X, S=0, A=a}  \nonumber \\
    = &\E{\underbrace{\frac{\ind{\Delta=1} Y}{\bar{G}_{s=0,a}(Y \lvert X)}}_{\tn{Term 1}} + \underbrace{\frac{\ind{\Delta=0} Q_{s=0, a}(X,C)}{\bar{G}_{s=0,a}(C \lvert X)}}_{\tn{Term 2}} - \underbrace{\int_{-\infty}^{\tilde{Y}} \frac{Q_{s=0, a}(X, c)}{\bar{G}_{s=0,a}^2(c \lvert X)} dG_{s=0,a}(c \lvert X)}_{\tn{Term 3}} \Big\lvert X,S=0,A=a} \label{eq:oct3-1}
\end{align}
where in the last line, we have substituted the definition of $\psi^{*}_{s=0, a}$ from~\cref{eq:def_rubin_signal}.
We will proceed separately for each term. Note that $\ind{\Delta=1} = \ind{Y \leq C}$.
\begin{align*}
    &\E{\frac{\ind{\Delta=1} Y}{\bar{G}_{s=0,a}(Y \lvert X)} \Big\lvert X,S=0,A=a} \\
    =& \E{\frac{Y}{\bar{G}_{s=0,a}(Y \lvert X)} \Big\lvert X,S=0,A=a, Y \leq C} P(Y \leq C \lvert X,S=0,A=a) \\
    =&\sum_y \frac{y}{\bar{G}_{s=0,a}(y \lvert X)} P(Y=y \lvert X,S=0,A=a, Y \leq C) P(Y \leq C \lvert X,S=0,A=a) \\
    =&\sum_y \frac{y}{\bar{G}_{s=0,a}(y \lvert X)} P(Y \leq C  \lvert Y=y, X,S=0,A=a) P(Y=y \lvert X,S=0,A=a) \\
    =&\sum_y \frac{y}{\bar{G}_{s=0,a}(y \lvert X)} P(y \leq C  \lvert Y=y, X,S=0,A=a) P(Y(a)=y \lvert X,S=0) \tag{Lemma~\ref{suplem:CateId}} \\
    =&\sum_y \frac{y}{\bar{G}_{s=1,a}(y \lvert X)} P(y \leq C  \lvert Y=y, X,S=1,A=a) P(Y(a)=y \lvert X,S=1) \tag{Lemma~\ref{suplem:gclem0}, Assumptions~\ref{asm:evo}, \ref{asm:gc}} \\
    =&\E{\frac{\ind{\Delta=1} Y}{\bar{G}_{s=1,a}(Y \lvert X)} \Big\lvert X,S=1,A=a} \tag{by symmetry}
\end{align*}
We continue with Term 2
\begin{align*}
    &\E{\frac{\ind{\Delta=0} Q_{s=0,a} (X,C)}{\bar{G}_{s=0,a}(C \lvert X)} \Big\lvert X,S=0,A=a} \\
    =&\E{\frac{Q_{s=0,a} (X,C)}{\bar{G}_{s=0,a}(C \lvert X)} \Big\lvert X,S=0,A=a,Y>C} P(Y > C \lvert X,S=0,A=a) \\
    =&\sum_c \frac{Q_{s=0,a} (X,c)}{\bar{G}_{s=0,a}(c \lvert X)} P(C=c \lvert X,S=0,A=a,Y>C) P(Y > C \lvert X,S=0,A=a) \\ 
    =&\sum_c \frac{Q_{s=1,a} (X,c)}{\bar{G}_{s=1,a}(c \lvert X)} P(C=c, Y>C \lvert X,S=0,A=a) \tag{Lemma~\ref{suplem:gclem0}}\\
    =&\sum_c \frac{Q_{s=1,a} (X,c)}{\bar{G}_{s=1,a}(c \lvert X)} \sum_y P(C=c, Y>C \lvert Y=y, X,S=0,A=a) P(Y=y \lvert X,S=0,A=a) \\
    =&\sum_c \frac{Q_{s=1,a} (X,c)}{\bar{G}_{s=1,a}(c \lvert X)} \sum_y P(C=c, y>C \lvert Y=y, X,S=0,A=a) P(Y(a)=y \lvert X,S=0) \tag{Lemma \ref{suplem:CateId}} \\
    =&\sum_c \frac{Q_{s=1,a} (X,c)}{\bar{G}_{s=1,a}(c \lvert X)} \sum_y P(C=c, y>C \lvert Y=y, X,S=1,A=a) P(Y(a)=y \lvert X,S=1) \tag{Assumptions~\ref{asm:evo}, \ref{asm:gc} } \\
    =&\E{\frac{\ind{\Delta=0} Q_{s=1,a} (X,C)}{\bar{G}_{s=1,a}(C \lvert X)} \Big\lvert X,S=1,A=a} \tag{by symmetry}
\end{align*}
We continue with Term 3
\begin{align*}
    &\E{\int_{-\infty}^{\tilde{Y}} \frac{Q_{s=0, a}(X, c)}{\bar{G}_{s=0,a}^2(c \lvert X)} dG_{s=0,a}(c \lvert X)\Big\lvert X,S=0,A=a} \\
    &=\sum_{y,c} \mathbb{E} \bigg[\int_{-\infty}^{\tilde{y}} \frac{Q_{s=0, a}(X, c)}{\bar{G}_{s=0,a}^2(c \lvert X)} dG_{s=0,a}(c \lvert X)\Big\lvert Y=y,C=c,X,S=0,A=a\bigg] P(Y=y,C=c \lvert X,S=0,A=a) \\
    &=\sum_{y,c} \left( \int_{-\infty}^{\tilde{y}} \frac{Q_{s=0, a}(X, c)}{\bar{G}_{s=0,a}^2(c \lvert X)} dG_{s=0,a}(c \lvert X) \right) P(C=c \lvert Y=y, X,S=0,A=a) P(Y(a)=y \lvert X,S=0) \\
    &=\sum_{y,c} \left( \int_{-\infty}^{\tilde{y}} \frac{Q_{s=1, a}(X, c)}{\bar{G}_{s=1,a}^2(c \lvert X)} dG_{s=1,a}(c \lvert X) \right) P(C=c \lvert Y=y, X,S=1,A=a) P(Y(a)=y \lvert X,S=1) \\
    &= \E{\int_{-\infty}^{\tilde{Y}} \frac{Q_{s=1, a}(X, c)}{\bar{G}_{s=1,a}^2(c \lvert X)} dG_{s=1,a}(c \lvert X)\Big\lvert X,S=1,A=a} \tag{by symmetry}
\end{align*}
where the second equality is by Lemma~\ref{suplem:CateId} and the third by Lemma~\ref{suplem:gclem0} and Assumptions~\ref{asm:evo} and \ref{asm:gc}.
\end{proof}
\subsection{Theorem 5}
\gcThmtwo*
\begin{proof}
    We write the following $\forall a \in \cbrc{0,1}$
    \begin{align}
    &\E{\frac{\ind{S=0,A=a} \tilde{Y}}{P(S=0,A=a \lvert X)} \Big\lvert X} \nonumber \\ 
        &= \E{\tilde{Y} \lvert X, S=0, A=a} \nonumber \\
        &= \E{\min (Y,C) \lvert X, S=0, A=a} \nonumber \\
        &= \E{Y \lvert X,S=0,A=a,Y\leq C} P(Y\leq C \lvert X,S=0,A=a) \nonumber \\
        &\hspace{30pt} + \E{C \lvert X,S=0,A=a,Y>C} P(Y>C \lvert X,S=0,A=a) \nonumber  \\ 
        &= \E{Y \lvert X,S=1,A=a,Y\leq C} P(Y\leq C \lvert X,S=1,A=a) \label{eq:sep16-2} \\
        &\hspace{30pt} + \E{C \lvert X,S=1,A=a,Y>C} P(Y>C \lvert X,S=1,A=a) \tag{Lemmas~\ref{suplem:gclem1} and \ref{suplem:gclem2}} \nonumber \\ 
        &=\E{\frac{\ind{S=1,A=a} \tilde{Y}}{P(S=1,A=a \lvert X)} \Big\lvert X} \tag{by symmetry}
    \end{align}
    which immediately gives us
    \begin{equation*}
        \E{\psi^{\tn{IPW}, \tilde{Y}}_{s=0} \lvert X} = \E{\psi^{\tn{IPW}, \tilde{Y}}_{s=1} \lvert X}
    \end{equation*}
    by definition of $\psi_{s}^{\tn{IPW}, \tilde{Y}}$ and we are done.
\end{proof}
\subsection{Proposition 2}
\propNoExMepo*
\begin{proof}
In the first part of this proof, where we show that \eqref{eq:gcipwthm} is not true in general, we will consider the following distributions for the potential outcomes which satisfy the \enquote{mean exchangeability of the potential outcomes} (the condition given in the statement of the proposition) but not the ignorability of selection condition in \Cref{asm:evo}
\begin{equation} \label{eq:poDistL4}
\begin{aligned}
    P(Y(0) = n \lvert S=0) &=
    \begin{dcases}
        1/5 & n \in \cbrc{1,2,3,4,5} \\
        0 & \tn{otherwise}
    \end{dcases} \\
    P(Y(1) = n \lvert S=0) &=
    \begin{dcases}
        1/5 & n \in \cbrc{2,3,4,5,6} \\
        0 & \tn{otherwise}
    \end{dcases} \\
    P(Y(0) = n \lvert S=1) &=
    \begin{dcases}
        1/3 & n \in \cbrc{2,3,4} \\
        0 & \tn{otherwise}
    \end{dcases} \\
    P(Y(1) = n \lvert S=1) &=
    \begin{dcases}
        1/6 & n = 2 \\
        1/2 & n = 4 \\
        1/3 & n = 5 \\
        0 & \tn{otherwise}
    \end{dcases} \\
\end{aligned}
\end{equation}
and the following {\em conditional} distributions for the censoring time
\begin{equation} \label{eq:cenDistL4}
\begin{aligned}
        P(C = n \lvert Y < 3.5) &=
    \begin{dcases}
        1 & n = 10 \\
        0 & \tn{otherwise}
    \end{dcases} \\
    P(C = n \lvert Y > 3.5) &=
    \begin{dcases}
        1 & n = 1/2 \\
        0 & \tn{otherwise}
    \end{dcases} \\
\end{aligned}
\end{equation}
For simplicity, we further assume that 
\begin{equation} \label{eq:L4fura}
\begin{aligned}
    Y(0), Y&(1) \indep X, A \lvert S \\
    C & \indep X,S,A \lvert Y
\end{aligned}
\end{equation}
Note that \eqref{eq:cenDistL4} and \eqref{eq:L4fura} are stricter versions of no unobserved confounding (\Cref{asm:iva}) and global censoring (Assumption~\ref{asm:gc}). We then note, by \eqref{eq:poDistL4}, the following
\begin{equation} \label{eq:meanPOL4}
\begin{aligned}
    \E{Y(0) \lvert X, S=0} &= 3 \\ 
    \E{Y(1) \lvert X, S=0} &= 4 \\ 
    \E{Y(0) \lvert X, S=1} &= 3 \\ 
    \E{Y(1) \lvert X, S=1} &= 4 \\ 
\end{aligned}
\end{equation}
That is, mean exchangeability of the potential outcomes $\E{Y(a) \lvert X,S=0} = \E{Y(a) \lvert X,S=1}$, $\forall a \in \cbrc{0,1}$ holds. Note that, however, ignorability of selection (in \Cref{asm:evo}) is violated (see \eqref{eq:poDistL4}). Thanks to \eqref{eq:L4fura}, we can simply calculate
\begin{equation} \label{eq:L4interprobcalc}
\begin{aligned}
    P(Y\leq C \lvert X, S=0, A=0) &= \frac{3}{5} \\
    P(Y\leq C \lvert X, S=0, A=1) &= \frac{2}{5} \\
    P(Y\leq C \lvert X, S=1, A=0) &= \frac{2}{3} \\
    P(Y\leq C \lvert X, S=1, A=1) &= \frac{1}{6}
\end{aligned}
\end{equation}
since, {\em e.g.}, when $S=0$ and $A=0$, we have $Y \leq C$ if and only if $Y(0) \in \{1,2,3\}$, which happens with probability $3 \times 1/5 = 3/5$. Next, we have
\begin{equation} \label{eq:L4means}
\begin{aligned}
    \E{Y \lvert X, S=0, A=0, Y\leq C} &= \frac{1 + 2 + 3}{3} = 2 \\
    \E{Y \lvert X, S=0, A=1, Y\leq C} &= \frac{2 + 3}{2} = 5/2 \\
    \E{Y \lvert X, S=1, A=0, Y\leq C} &= \frac{2 + 3}{2} = 5/2 \\
    \E{Y \lvert X, S=1, A=1, Y\leq C} &= 2
\end{aligned}
\end{equation}
and
\begin{equation} \label{eq:L4means_C}
    \E{C \lvert X,S,A, Y>C} = \frac{1}{2}
\end{equation}
since $C = 1/2$ almost surely whenever $Y>C$. We are now ready to calculate 
\begin{align*}
    \E{\psi^{\tn{IPW}, \tilde{Y}} \lvert X} 
    &=\E{ \left(\psi^{\tn{IPW}, \tilde{Y}}_{s=1,a=1} - \psi^{\tn{IPW}, \tilde{Y}}_{s=1,a=0} \right) - \left( \psi^{\tn{IPW}, \tilde{Y}}_{s=0,a=1} - \psi^{\tn{IPW}, \tilde{Y}}_{s=0,a=0} \right) \bigg\lvert X}
\end{align*}
where
\[
\psi^{\tn{IPW}, \tilde{Y}}_{s,a} = \frac{\ind{S=s,A=a} \tilde{Y}}{P(S=s,A=a \lvert X)}
\]
From \eqref{eq:sep16-2}, we have
\begin{align*}
    \E{\psi^{\tn{IPW}}_{S,A} \lvert X} = \E{Y \lvert X,S,A,Y\leq C} P(Y\leq C \lvert X,S,A) + \E{C \lvert X,S,A,Y>C} P(Y>C \lvert X,S,A)
\end{align*}
Plugging in the values calculated in \eqref{eq:poDistL4}, \eqref{eq:cenDistL4}, \eqref{eq:L4means}, \eqref{eq:L4means_C}
\begin{equation} \label{eq:ipw-impute-means}
\begin{aligned}
    \E{\psi^{\tn{IPW}, \tilde{Y}}_{s=0,a=0} \lvert X} &= \frac{3}{5} \cdot 2 + \frac{2}{5} \cdot \frac{1}{2} = \frac{7}{5} \\ 
    \E{\psi^{\tn{IPW}, \tilde{Y}}_{s=0,a=1} \lvert X} &= \frac{2}{5} \cdot \frac{5}{2} + \frac{3}{5} \cdot \frac{1}{2} = \frac{13}{10} \\
    \E{\psi^{\tn{IPW}, \tilde{Y}}_{s=1,a=0} \lvert X} &= \frac{2}{3} \cdot \frac{5}{2} + \frac{1}{3} \cdot \frac{1}{2} = \frac{11}{6} \\
    \E{\psi^{\tn{IPW}, \tilde{Y}}_{s=1,a=1} \lvert X} &= \frac{1}{6} \cdot 2 + \frac{5}{6} \cdot \frac{1}{2} = \frac{3}{4} \\
\end{aligned}
\end{equation}
We are done since
\begin{align*}
    \E{\psi^{\tn{IPW}, \tilde{Y}} \lvert X} = \left(\frac{3}{4} - \frac{11}{6} \right) - \left(\frac{13}{10} - \frac{7}{5} \right) = -\frac{59}{60} \neq 0 
\end{align*}
Next, we show that \eqref{eq:gccdrthm} is not true in general. We keep the same setup above, only changing the marginal distributions of the potential outcomes to the following for simplicity
\begin{equation} \label{eq:poDistL5}
\begin{aligned}
    P(Y(0) = n \lvert S=0) &=
    \begin{dcases}
        1 & n = 0 \\
        0 & \tn{otherwise}
    \end{dcases} \\
    P(Y(1) = n \lvert S=0) &=
    \begin{dcases}
        1 & n = 2 \\
        0 & \tn{otherwise}
    \end{dcases} \\
    P(Y(0) = n \lvert S=1) &=
    \begin{dcases}
        1 & n = 0 \\
        0 & \tn{otherwise}
    \end{dcases} \\
    P(Y(1) = n \lvert S=1) &=
    \begin{dcases}
        \frac{1}{2} & n \in \cbrc{0, 4} \\
        0 & \tn{otherwise}
    \end{dcases} \\
\end{aligned}
\end{equation}
Note that the mean exchangeability of potential outcomes again hold, but not the ignorability of selection in \Cref{asm:evo}.

When $(s=0,a=0)$, $(s=0,a=1)$, or $(s=1,a=0)$, we have $C=10$ almost surely, since $Y<3.5$ almost surely and we never have censored outcomes (see \eqref{eq:cenDistL4}). We have $\tilde{Y} = Y$ and $\Delta = 1$ almost surely. Then for those values of $S=s$ and $A=a$, we can write, by \eqref{eq:oct3-1}
\begin{align*}
    &\E{\psi_{s,a}^{\tn{CDR}} \lvert X}  \\
    = &\E{\frac{\ind{\Delta=1} Y}{\bar{G}_{s,a}(Y \lvert X)} + \frac{\ind{\Delta=0} Q_{s, a}(X,C)}{\bar{G}_{s,a}(C \lvert X)} - \int_{-\infty}^{\tilde{Y}} \frac{Q_{s, a}(X, c)}{\bar{G}_{s,a}^2(c \lvert X)} dG_{s,a}(c \lvert X) \Big\lvert X,S=s,A=a} \\
    = &\E{\frac{Y}{\bar{G}_{s,a}(Y \lvert X)} - \int_{-\infty}^{Y} \frac{Q_{s, a}(X, c)}{\bar{G}_{s,a}^2(c \lvert X)} \delta(c-10) \Big\lvert X,S=s,A=a} \\
    = &\E{\frac{Y}{\bar{G}_{s,a}(Y \lvert X)} \Big\lvert X,S=s,A=a} \\
\end{align*}
where $\delta$ is the Dirac delta function. The integral term disappears since $Y < 10$ almost surely. Since $\bar{G}_{s,a}(t \lvert X) = 1$ for all $t<10$, we can calculate the following
\begin{align*}
    \E{\psi^{\tn{CDR}}_{s=0,a=0} \lvert X} &= 0 \\
    \E{\psi^{\tn{CDR}}_{s=0,a=1} \lvert X} &= 2 \\
    \E{\psi^{\tn{CDR}}_{s=1,a=0} \lvert X} &= 0 
\end{align*}
Next, we calculate $\E{\psi_{s=1,a=1}^{\tn{CDR}} \lvert X}$. Note that in this case, the observations are censored with probability 1/2 (when $Y=4$) and not censored with probability 1/2 (when $Y=0$). We then note the marginal (w.r.t. $Y$) survival function of the censoring time as  
\begin{equation*}
\begin{aligned}
    \bar{G}_{s=1,a=1} (n \lvert X) &=
    \begin{dcases}
        1 & n < \frac{1}{2} \\
        \frac{1}{2} & \frac{1}{2} \leq n < 10 \\
        0 & \tn{otherwise}
    \end{dcases} 
\end{aligned}
\end{equation*}
And the density function as the sum of two Dirac delta functions
\begin{equation} \label{eq:oct3-2}
    dG_{s=1,a=1} (n \lvert X) = \frac{1}{2} \left( \delta(n-\frac{1}{2}) + \delta(n-10) \right)
\end{equation}
We can then write, also utilizing Lemma~\ref{suplem:CateId} and \eqref{eq:poDistL5}
\begin{align*}
    &\E{\psi_{s=1,a=1}^{\tn{CDR}} \lvert X}  \\
    = &\E{\frac{\ind{\Delta=1} Y}{\bar{G}_{s=1,a=1}(Y \lvert X)} + \frac{\ind{\Delta=0} Q_{s=1, a=1}(X,C)}{\bar{G}_{s=1,a=1}(C \lvert X)} - \int_{-\infty}^{\tilde{Y}} \frac{Q_{s=1, a=1}(X, c)}{\bar{G}_{s=1,a=1}^2(c \lvert X)} dG_{s=1,a=1}(c \lvert X) \Big\lvert X,S=1,A=1} \\
    = &\E{\frac{\ind{\Delta=1} Y}{\bar{G}_{s=1,a=1}(Y \lvert X)} + \frac{\ind{\Delta=0} Q_{s=1, a=1}(X,C)}{\bar{G}_{s=1,a=1}(C \lvert X)} - \int_{-\infty}^{\tilde{Y}} \frac{Q_{s=1, a=1}(X, c)}{\bar{G}_{s=1,a=1}^2(c \lvert X)} dG_{s=1,a=1}(c \lvert X) \Big\lvert X,Y=0,S=1,A=1} \\ &\hspace{20pt}\times P(Y(1)=0 \lvert X,S=1) \\
    +&\E{\frac{\ind{\Delta=1} Y}{\bar{G}_{s=1,a=1}(Y \lvert X)} + \frac{\ind{\Delta=0} Q_{s=1, a=1}(X,C)}{\bar{G}_{s=1,a=1}(C \lvert X)} - \int_{-\infty}^{\tilde{Y}} \frac{Q_{s=1, a=1}(X, c)}{\bar{G}_{s=1,a=1}^2(c \lvert X)} dG_{s=1,a=1}(c \lvert X) \Big\lvert X,Y=4,S=1,A=1} \\ 
    &\hspace{20pt}\times P(Y(1)=4 \lvert X,S=1) \\
    = &\underbrace{\E{\frac{Y}{\bar{G}_{s=1,a=1}(Y \lvert X)} \Big\lvert X,Y=0,S=1,A=1}}_{=0} \times \frac{1}{2} \tag{when $Y=0$, we have $C=10$, $\Delta = 1, \tilde{Y} = 0$} \\
    +&\E{\frac{Q_{s=1, a=1}(X,C)}{\bar{G}_{s=1,a=1}(C \lvert X)} - \int_{-\infty}^{\tilde{Y}} \frac{Q_{s=1, a=1}(X, c)}{\bar{G}_{s=1,a=1}^2(c \lvert X)} dG_{s=1,a=1}(c \lvert X) \Big\lvert X,Y=4,S=1,A=1} \times \frac{1}{2} \tag{when $Y=4$, we have $C=\frac{1}{2}$, $\Delta = 0, \tilde{Y} = \frac{1}{2}$} \\
    =&\left( \frac{Q_{s=1, a=1}(X,\frac{1}{2})}{\bar{G}_{s=1,a=1}(\frac{1}{2} \lvert X)} - \int_{-\infty}^{\frac{1}{2}} \frac{Q_{s=1, a=1}(X, c)}{\bar{G}_{s=1,a=1}^2(c \lvert X)} dG_{s=1,a=1}(c \lvert X) \right) \times \frac{1}{2}  \\
    =&\left( \frac{Q_{s=1, a=1}(X,\frac{1}{2})}{\bar{G}_{s=1,a=1}(\frac{1}{2} \lvert X)} - \frac{1}{2}\left(\frac{Q_{s=1, a=1}(X, \frac{1}{2})}{\bar{G}_{s=1,a=1}^2(\frac{1}{2} \lvert X)} \right) \right) \times \frac{1}{2} \tag{by \eqref{eq:oct3-2}} \\
    =&\left( \frac{4}{\frac{1}{2}} - \frac{1}{2}\frac{4}{\left(\frac{1}{2}\right)^2} \right) \times \frac{1}{2}\\
    =&0
\end{align*}
Note that $Q_{s=1, a=1}(X,\frac{1}{2}) = \E{Y \lvert X,Y>\frac{1}{2},S=1,A=1} = 4$ since given $Y>\frac{1}{2}$, $Y=4$ almost surely. We are done since
\begin{align*}
    \E{\psi^{\tn{CDR}} \lvert X} &=\E{ \left(\psi^{\tn{CDR}}_{s=1,a=1} - \psi^{\tn{CDR}}_{s=1,a=0} \right) - \left( \psi^{\tn{CDR}}_{s=0,a=1} - \psi^{\tn{CDR}}_{s=0,a=0} \right) \bigg\lvert X}\\
    &=\left(0 - 0 \right) - \left(2 - 0 \right) = -2 \neq 0 
\end{align*}
\end{proof}
\section{DETAILS ON THE DATASETS USED IN THE EXPERIMENTS}
\label{app:datasets}
\subsection{IHDP Experiments}
\label{app:ihdp}
We use 10 covariates from the IHDP data for $X$, containing both continuous and categorical variables: \texttt{twin}, \texttt{b.head}, \texttt{preterm}, \texttt{momage}, \texttt{bw}, \texttt{b.marr}, \texttt{nnhealth}, \texttt{birth.o}, \texttt{momhisp}, \texttt{sex}. 

\subsubsection{Propensity Score of the Treatment} In the RCT cohort $S=0$ we set $P(A=1 \lvert X, S=0) = 0.5$. That is, the treatment assignment is completely randomized. In the OS cohort, we have
\begin{align*}
    P(A=1 \lvert X,S=1) = \tn{sigmoid} (X^{\intercal} \beta_{\mathrm{prop}} + C)
\end{align*}
where $C \in \{0.5, 0.75, 1, 1.25\}$ and $\beta_{\mathrm{prop}} (\texttt{sex})  = 2 \times C$ depending on the \enquote{strength} of confounding we want to induce. We set $\beta_{\mathrm{prop}} (x) = 0$ for all $x \neq \texttt{sex}$. That is, among the 10 covariates listed above, only the \texttt{sex} covariate has any influence on the treatment assignment. See the \texttt{prop\_args} parameter in \texttt{exp\_configs/ihdp/*.json} for the specific value of $\beta_{\mathrm{prop}}$ in each experimental setup.

The propensity score of the treatment is used to sample a binary treatment $A_i$ for a patient with covariates $x_i$ in the study $s_i$ as 
\[
A_i \sim \texttt{Bernoulli} (P(A=1 \lvert X=x_i, S=s_i))
\]
\subsubsection{Survival Functions of Time-to-Event and Censoring Time} 
Time-to-events $Y$ and censoring times $C$ are simulated using a CoxPH model \citep{cox1972regression} for their survival functions $\bar{F}_{s,a} (t \lvert X)$ \eqref{eq:fbar} and $\bar{G}_{s,a} (t \lvert X)$ \eqref{eq:gbar}, which admits the general form
\begin{equation*}
    \overline{W}_0 (t; \lambda, p)^{\exp(X^{\intercal} \beta_{\mathrm{Cox}})}
\end{equation*}
where
\begin{equation*}
    \overline{W}_0 (t; \lambda, p) =  \exp\left(-(\lambda t)^p \right) \qquad \lambda, p \in \mathbb{R}_+
\end{equation*}
Recall that the parameters $\lambda$, $p$, and $\beta_{\mathrm{Cox}} \in \mathbb{R}^{10}$ are specified separately for each study $S \in \{0,1\}$ and treatment group $A \in \{0,1\}$ pair. The specific values for each parameter can be found in $\texttt{exp\_configs/ihdp/*.json}$ under the corresponding study key (RCT or OS).

To measure type-1 error (see setup \#1 in \Cref{table:cic_synth_exp}), we use identical values for the parameters across the RCT and the OS. Identical parameters are also used in the experiments where we consider the violation of internal validity (see setups \#4 and \#5 in \Cref{table:cic_synth_exp}, and middle column in \Cref{fig:cicsweep}). Only difference is that we conceal the $\texttt{sex}$ covariate to induce unmeasured confounding. Finally, to simulate the violation of the external validity  (see setups \#2 and \#3 in \Cref{table:cic_synth_exp}, and left column in \Cref{fig:cicsweep}), we use different values for $\beta_{\mathrm{Cox}} (\texttt{nnhealth})$ in the treatment $A=1$ groups of the RCT $S=0$ and the OS $S=1$. We experiment with different magnitudes of difference for the parameter, denoted by $\Delta \beta_{\mathrm{Cox}}$ in the main text ({\em e.g.}, see \Cref{table:cic_synth_exp}).
\subsubsection{Estimation of the Nuisance Functions} 
Since we know the exact data-generating process, we are able to correctly specify the models for estimation.

For the selection $P(S=1 \lvert X)$ and propensity $P(A=1 \lvert X,S)$ scores, we fit logistic regression models. To limit the variance of the signals, we drop patients with extreme selection or propensity scores ({\em i.e.}, < 0.05 or > 0.95). For the survival functions $\bar{F}_{s,a} (t \lvert X)$ \eqref{eq:fbar} and $\bar{G}_{s,a} (t \lvert X)$ \eqref{eq:gbar}, we fit a CoxPH model where the baseline survival function $\overline{W}_0 (t; \lambda, p)$ is estimated via the Breslow estimator and $\beta_{\mathrm{Cox}}$ is estimated by fitting Cox's partial likelihood \citep{cox1972regression, davidson2019lifelines}. Finally, we fit an XGboost model for the mean outcome regressors as part of the DR-$\tilde{Y}$ and DR-$Y$ signals (see \Cref{app:baseline_signals}).
\subsection{WHI Experiments}
\label{app:whi}
WHI data is available to all researchers upon request in \url{https://biolincc.nhlbi.nih.gov/studies/whi_ctos/}. We start with 1121 features available at the baseline for both the RCT and the OS cohorts. After removing duplicate and highly correlated (Pearson coefficient > 0.95), we had 954 features. Finally, we transform to 350 principal components (PC), which capture 90\% of the variance, and use those PCs as the set of $X$ variables. THE PC transformation helps with the convergence and instability-related issues in the generalized linear models, specifically the Cox regression model. The estimation of the nuisance functions is done the same way as in the IHDP experiments (see \Cref{app:ihdp}).
\section{ADDITIONAL EXPERIMENTAL RESULTS}
\label{app:experiments_add}
\subsection{Testing the Doubly-Robustness of the CDR Signal} \label{supssec:drtest}
We test the CDR signal's doubly-robust property for setups \#1,2,3 in \Cref{table:cic_synth_exp} where internal validity holds. 

We misspecify $\bar{F}_{S,A} (t \lvert X)$ by setting the baseline survival function $\overline{W}_0 (t; \lambda, p)$ to follow the law of a uniform random variable between the minimum and maximum time-to-event variables, and set $\beta_{\mathrm{Cox}} = 0$. We use the correct (oracle) values for $\bar{G}_{S,A} (t \lvert X)$ and $P(A \lvert X,S)$, and name the resulting signal \texttt{CDR-MissF}.

We then misspecify $\bar{G}_{S,A} (t \lvert X)$ and $P(A \lvert X,S)$ and use the correct values for $\bar{F}_{S,A} (t \lvert X)$, and name the resulting signal \texttt{CDR-MissGP}. The misspecification of $\bar{G}_{S,A} (t \lvert X)$ is done similarly to that of $\bar{F}_{S,A} (t \lvert X)$, by using the minimum and maximum censoring times. We also misspecify $P(A \lvert X,S)$ by setting it to $0.5$ in the OS $S=1$. 

We see in \Cref{table:cdr} that misspecified models maintain low type-1 error and higher power, corroborating the doubly-robustness of the CDR signal.
\begin{table}[h]
\centering

\caption{Rejection rates over 40 runs for misspecified CDR signals. OS size is $n_1 = 2955$. Same setups in \Cref{table:cic_synth_exp} are considered.}

\label{table:cdr}

\begin{tabular}{lccc}
\toprule
\toprule
\textbf{Setup \#} & \textbf{1} & \textbf{2} & \textbf{3} \\ 
\midrule
\textbf{Metric} & \cred{\textbf{Type-1 error}} & \cblue{\textbf{Power}} & \cblue{\textbf{Power}} \\ 
\midrule
CDR-MissF   & 0 & 0.3 & 0.975     \\ 
CDR-MissGP & 0 & 0.35 & 0.975 \\

\bottomrule
\bottomrule
\end{tabular}
\end{table}

\subsection{Witness Function} \label{supssec:witfunc}
As exposed in \citet{hussain2023falsification}, an appealing feature of using MMR-based approach is that we can express the maximizer of $\mathbb{M}$ (\Cref{thm:thmMMR}) as
\begin{align*}
f^* = \arg \sup_{f \in {\cal F}, ||f|| \le 1}\left(\E{\psi f(X)}\right)^2
\end{align*}
in closed form. This maximizer is referred to as the \emph{witness function} and can be estimated as follows:
\begin{align*}
    \hat{f}^{*} = C \frac{1}{n} \sum_i \psi_i k(x_i,x)   
\end{align*}
where $C$ is a constant such that $\int_{\mathcal{X}} f^*(x) dx = 1$.

The witness function reveals the regions of ${\cal X}$ where the \enquote{difference signal} $\psi$ takes on larger values. As such, it can indicate the sub-populations where the CATE (see \eqref{eq:cate}) estimates from the RCT and OS are most discrepant. 

We investigated the witness function using a fully synthetic dataset, where we control the \enquote{source covariates} of discrepancy and ensure limited correlation between them to facilitate the validation of the readings.

For both RCT $S=0$ and OS $S=1$, we generate covariates using one intercept and 10 independent normal variables with variances $\sigma^2 = 1$ and the following means
\begin{align*}
    \mu_{s=0} &= \begin{bmatrix} 0,0,0,0,0,0,0,0,1,0 \end{bmatrix}\\
    \mu_{s=1} &= \begin{bmatrix} 0, -0.1, 0.4, 0, -0.3, 0.15, 0, 0.4, 1, -0.4 \end{bmatrix}\\
\end{align*}
\begin{figure}[htbp]
    \centering
    \includegraphics[width=0.9\linewidth]{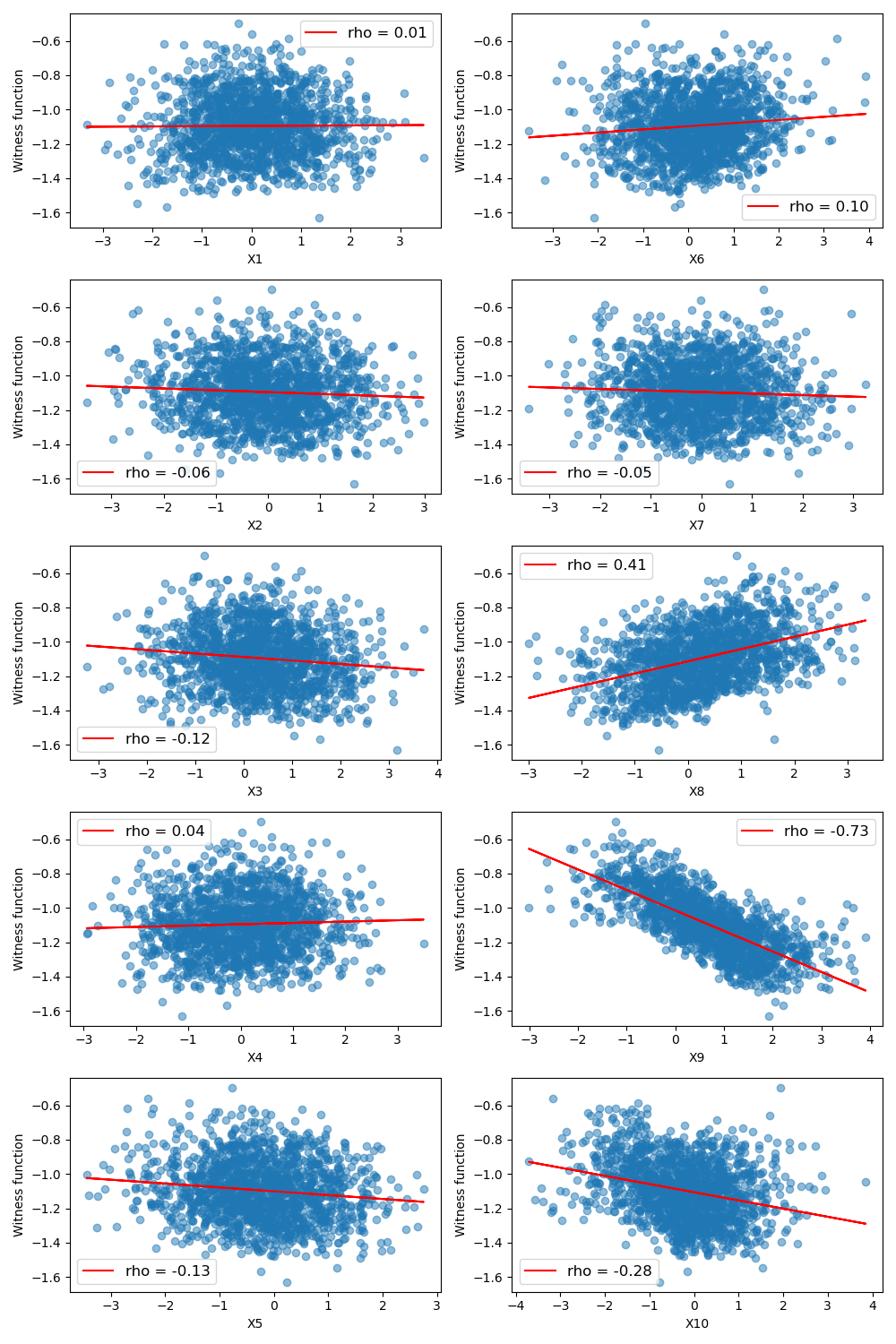}
    \caption{The difference signal $\psi$ (blue) with a line fit to it (red) separately for each covariate.}
    \label{fig:witness_indiv}
\end{figure}
For the OS, we generated the probability of treatment with a logistic regression with parameters $\beta_{\mathrm{prop}}$:
\begin{align*}
    P(A=1 \mid X, S=1) &= \tn{sigmoid}(X^{\intercal} \beta_{\mathrm{prop}})\\
     \tn{where} \quad \beta_{\mathrm{prop}, s=1} &= \begin{bmatrix} -0.7, 0.4, -0.2, 0.3, -0.1, -0.4, 0.2, 0.1, 0.4, -0.8,-0.75\end{bmatrix}
\end{align*}
and for the RCT we have $P(A=1 \lvert X,S=0) =0.5$. We induce different CATE functions in the RCT and OS by using different parameter values for covariates $X_8$, $X_9$, and $X_{10}$ in the CoxPH model for generating $Y(1)$.
\begin{align*}
    \beta_{\mathrm{Cox},S=0,Y(1)} &= \begin{bmatrix} 0, 0.7, -0.4, 0.5, 0.4, -0.5, 0.6, -0.4, 0.5, -1.2, -0.7 \end{bmatrix}.\\
    \beta_{\mathrm{Cox},S=1,Y(1)} &= \begin{bmatrix} 0, 0.7, -0.4, 0.5, 0.4, -0.5, 0.6, -0.4, -0.5, 1.2, 0.7 \end{bmatrix}\\
\end{align*}
The CoxPH model parameters for the potential outcome $Y(0)$ is set to be the same across studies. According to this data generating model, an effective witness function should enable the detection as $X_8$, $X_9$, and $X_{10}$ as culprits for the discrepancy between RCT and OS.

In \Cref{fig:witness_indiv}, we scatter-plot (blue) the individual values of the difference signal $\psi$. The witness function at any point is then a weighted average of those values depending on the specific kernel function $k(\cdot,\cdot)$. We visualize the linear fit (red) to the difference signal function $\psi$ over $X$ for some quick insight. We observe the strongest correlations for variables $X_8$, $X_9$, and $X_{10}$, as expected.

In Figure~\ref{fig:witness_all}, we plot the witness function over each dimension individually (using the same values for the other dimensions, effectively excluding their effect on the result). We repeat the experiment with some additive noise on the corresponding covariate values to obtain uncertainty ranges. We observe that large $X_8$, $X_9$, and $X_{10}$ result in high witness function values, pronouncing the difference in the corresponding $\beta_{\mathrm{Cox}}$ values across studies. The direction of growth indicates the sign of the discrepancy. These experiments confirm that the witness function can be used for identifying a sub-population that exhibits a stronger violation of the validity assumptions. 

However, in practice, the covariates will be correlated, which can hamper the ability of the witness function to pinpoint the variables responsible for the discrepancy. Furthermore, if the discrepancy results from unmeasured confounders, our only hope would be to see some effect through features correlated with the unmeasured confounders. However, useful proxies for the unmeasured confounders should also alleviate the very issue of confounding; therefore, the witness functions' utility is fundamentally limited under unmeasured confounding.
\begin{figure}[htbp]
    \centering
    \includegraphics[width=0.7\linewidth]{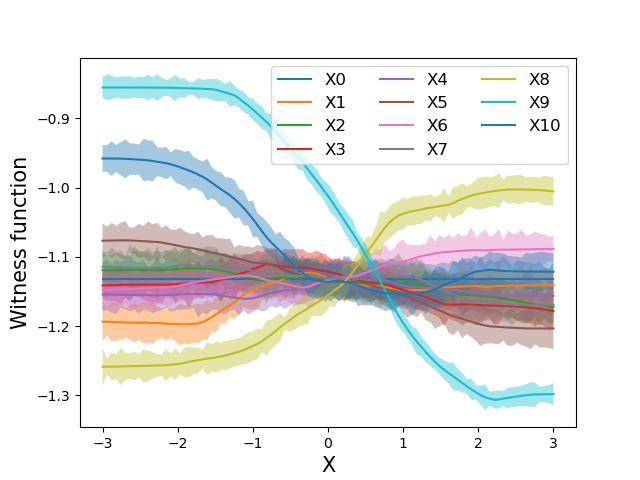}
    \caption{Witness function evaluated over each dimension individually.}
    \label{fig:witness_all}
\end{figure}
\subsection{Hypothesis Testing}
\label{app:test_details}
After calculating the test statistic $\hat{\mathbb{M}}_n^2$, we follow \citet{hussain2023falsification} (see their Appendix E.2) to generate $B=100$ samples from the null distribution $H_0$. Let $(w_{k1},\dots,w_{kn}) \sim \tn{Multinom}(n,(\frac{1}{n},\dots,\frac{1}{n}))$. The $k$-th bootstrap sample of the null is given as
\[
\hat{\mathbb{M}}_{n(k)}^2 = \frac{1}{n^2} \sum_{i,j \in {\cal I}, i \neq j} (w_{ki}-1)\hat{\psi}_i k(x_i, x_j)\hat{\psi}_j (w_{kj}-1)
\]

The $p$-value for the statistic is then calculated as 
\[
t_a = \frac{\left[\sum_{k=1}^{B} \mathbf{1}(\hat{\mathbb{M}}_n^2 \le \hat{\mathbb{M}}_{n(k)}^2)\right]+1}{B+1}
\]
If $t_a < 0.05$, we reject the null $H_0$ and accept it otherwise.

\section{BASELINE SIGNALS}
\label{app:baseline_signals}
We consider four signals without any component to model censoring that are used to construct the baseline tests in the ablation studies: IPW-Y, DR-Y, IPW-$\tilde{Y}$, and DR-$\tilde{Y}$. These signals are the standard inverse propensity weighting and regression-based signals in the literature. Let us start with the latter two. Note that $\psi^{\tn{IPW},\tilde{Y}}$ is already defined in \eqref{eq:ipw-impute-means}. We define
\begin{align*}
    \psi^{\tn{DR},\tilde{Y}}_{s,a} &\coloneqq \frac{1}{P(S=s\lvert X)}\left( \frac{\ind{A=a} \left(\tilde{Y} - \tilde{\mu}_{s,a} (X) \right)}{P(A=a \lvert X,S=s)} +  \tilde{\mu}_{s,a} (X) \right)  \\
    \psi^{\tn{DR},\tilde{Y}}_{s} &\coloneqq \psi^{\tn{DR},\tilde{Y}}_{s,a=1} - \psi^{\tn{DR},\tilde{Y}}_{s,a=0}  \\ 
    \psi^{\tn{DR},\tilde{Y}} &\coloneqq \psi^{\tn{DR},\tilde{Y}}_{s=1} - \psi^{\tn{DR},\tilde{Y}}_{s=0}
\end{align*}
where
\[
 \tilde{\mu}_{s,a} (X) = \E{\tilde{Y} \lvert X,S=s,A=a}
\]
is the mean outcome function for the {\em imputed outcome} $\tilde{Y}$. $\psi^{\tn{IPW},Y}$ and $\psi^{\tn{DR},Y}$ are the conjugates of $\psi^{\tn{IPW},\tilde{Y}}$ and $\psi^{\tn{DR},\tilde{Y}}$, invoked only on the uncensored data ($\Delta = 1$). Precisely,
\begin{align*}
    \psi^{\tn{DR},Y}_{s,a} &\coloneqq \frac{1}{P(S=s\lvert X, \Delta=1)}\left( \frac{\ind{A=a} \left(\tilde{Y} - \tilde{\mu}_{s,a} (X \lvert \Delta=1) \right)}{P(A=a \lvert X,S=s, \Delta=1)} +  \tilde{\mu}_{s,a} (X \lvert \Delta = 1) \right)  \\
    \psi^{\tn{DR},Y}_{s} &\coloneqq \psi^{\tn{DR},Y}_{s,a=1} - \psi^{\tn{DR},Y}_{s,a=0}  \\ 
    \psi^{\tn{DR},Y} &\coloneqq \psi^{\tn{DR},Y}_{s=1} - \psi^{\tn{DR},Y}_{s=0}
\end{align*}

where $\psi^{\tn{IPW},Y}$ is defined similarly and the signals are estimated using only the uncensored data.

\end{document}